%% file: spaa2021_arxiv.tex
\documentclass[letterpaper,11pt]{article}

\newcommand{\qedllncs}{}

\newif\ifincludefigures
\includefigurestrue

\usepackage[letterpaper, left=2.5cm, right=2.5cm, top=2.5cm, bottom=2.5cm]{geometry}
\usepackage{amsmath}
\usepackage{amsthm}
\usepackage{mathtools} 
\usepackage{booktabs} 
\usepackage{amsmath}
\usepackage{amssymb}
\usepackage{enumerate}
\usepackage{txfonts} 
\usepackage{url}
\usepackage{hyperref}
\usepackage{stackrel}
\usepackage{algorithm}
\usepackage{algpseudocode}
\usepackage{longtable}
\usepackage{tikz}
\usetikzlibrary{arrows.meta}

\usepackage[labelfont=bf]{caption}

\usepackage{array,ragged2e}
\newcolumntype{L}[1]{>{\RaggedRight\arraybackslash}p{#1}}
\newcolumntype{C}[1]{>{\centering\arraybackslash}p{#1}}
\newcolumntype{R}[1]{>{\RaggedLeft\arraybackslash}p{#1}}

\DeclareMathOperator*{\argmin}{argmin}

\newcommand\numberthis{\addtocounter{equation}{1}\tag{\theequation}} %
\newcommand{\alignstack}[2]{\stackrel{\mathmakebox[\widthof{\ensuremath{#2}}]{#1}}{#2}}
\renewcommand{\vec}[1]{\boldsymbol{#1}}

\algtext*{EndFor}
\algtext*{EndIf}

\newcounter{thm}
\newtheorem{theorem}[thm]{Theorem}
\newtheorem{lemma}[thm]{Lemma}
\newtheorem{corollary}[thm]{Corollary}

\begin{document}

\title{Algorithms for Right-Sizing Heterogeneous Data Centers\footnote{Work supported by the European Research Council, Grant Agreement No.\ 691672.}}

\author{Susanne Albers \\
	Technical University of Munich \\
	albers@in.tum.de \\
	\and
	Jens Quedenfeld\footnote{Contact author} \\
	Technical University of Munich \\
	jens.quedenfeld@in.tum.de \\
}

\maketitle

	\begin{abstract}
		
		Power consumption is a dominant and still growing cost factor in data centers. In time periods with low load, the energy consumption can be reduced by powering down unused servers. We resort to a model introduced by Lin, Wierman, Andrew and Thereska \cite{LinWierman2013,LinWierman2013extended} that considers data centers with identical machines, and generalize it to heterogeneous data centers with $d$ different server types. The operating cost of a server depends on its load and is modeled by an increasing, convex function for each server type. In contrast to earlier work, we consider the discrete setting, where the number of active servers must be integral. Thereby, we seek truly feasible solutions. For homogeneous data centers ($d=1$), both the offline and the online problem were solved optimally in  \cite{AlbersQuedenfeld2018,AlbersQuedenfeld2018extended}.
		
		In this paper, we study heterogeneous data centers with general time-dependent operating cost functions. We develop an online algorithm based on a work function approach which achieves a competitive ratio of $2d + 1 + \epsilon$ for any $\epsilon > 0$. For time-independent operating cost functions, the competitive ratio can be reduced to $2d + 1$.  There is a lower bound of $2d$ shown in \cite{AlbersQuedenfeld2021ciac}, so our algorithm is nearly optimal. For the offline version, we give a graph-based $(1+\epsilon)$-approximation algorithm. Additionally, our offline algorithm is able to handle time-variable data-center sizes.
		
	\end{abstract}
	
	\section{Introduction}
	\label{sec:intro}
	\input{introduction.tex}

	\section*{Notation}
	Let $[k] \coloneqq \{1, 2, \dots k\}$, $[k]_0 \coloneqq \{0, 1, \dots k\}$ and $[k:l] \coloneqq \{k, k+1, \dots, l\}$ where $k,l \in \mathbb{N}$. A tabular overview of the variables introduced in the following sections is shown in Appendix~\ref{sec:appendix:variables}.

	\section{Online Algorithm for time-\emph{independent} operating cost functions}
	\label{sec:online}
	\input{onlineAlgorithm.tex}

	\section{Online Algorithm for time-\emph{dependent} operating cost functions}
	\label{sec:online:time}

\input{onlineAlgorithmTime.tex}

	\section{Approximation Algorithm}
	\label{sec:approx}
	\input{approxAlgorithm.tex}

	\appendix

	\input{appendixVariables.tex}

	\bibliographystyle{plainurl}
	\bibliography{../../meta/literature}

\end{document}

%% file: introduction.tex




Energy conservation in data centers is important for both economical and ecological reasons \cite{Bawden2016}. 
A huge amount of the energy consumed in data centers is wasted
because many servers run idle for long time periods, while still consuming half of their peak
power \cite{Delforge2014,Schmid2009powerShortLink}. 
The power consumption can be reduced by powering down servers that are currently not needed.
However, a power-up operation of a server causes increased energy consumption. Hence, holding an idle server in active mode for a short period of time is cheaper than powering it down and up again shortly after.
Furthermore, power-up and -down operations generate delay and wear-and-tear costs \cite{LinWierman2013extended}. 
Therefore, algorithms are needed that dynamically right-size a data center depending on incoming jobs so as to minimize the energy consumption. 

In this paper, we consider data centers with heterogeneous servers. This can be different architectures, for example, servers that use the GPU to perform massive parallel calculations. However, tasks that contain many branches are not suitable for GPUs and can be processed much faster on a common CPU \cite{Shan2006}.
Heterogeneity may also result from old and new servers. It is a common practice that a data center is extended by new servers while the old ones are kept in use. 

In practice, the energy consumption of a server is not constant but increases with load~\cite{AndrewLinWierman2010}. If a machine is idle, the CPU frequency is lowered in modern hardware to save energy~\cite{Mittal2014}. For high frequencies, the CPU voltage has to be raised, which results in a superlinear increase in power consumption \cite{WiermanAndrewTang2009}. 
Therefore, in our model, the energy consumption of each server type $j$ is modeled by an increasing convex function $f_j$ of the load~$z$. The operating cost of an idle server is given by $f_j(0)$. By setting the value of $f_j$ to infinity for large load values $z$, it is possible to model different server capacities. 
For example, there may be a slow server type with a maximum load of $1$ and a fast server type with a maximum load of $4$ that can process four times as many jobs as the slow server. 

Our model described below is a generalization of the model presented by Lin, Wierman, Andrew and Thereska \cite{LinWierman2013,LinWierman2013extended} for homogeneous data centers where all servers are identical.

\textbf{Problem formulation.}
We consider a data center with $d$ different server types and $m_j$ servers of type $j$. The servers have two states, an active one where they are able to process jobs and an inactive one without energy consumption. Powering up a server of type~$j$, i.e., switching it from the inactive to the active state, produces cost of $\beta_j$ (called \emph{switching} cost). Power-down operations do not incur any cost. 
We consider a time horizon consisting of the time slots $\{1, \dots, T\}$. For each time slot $t \in \{1, \dots, T\}$, a job volume of $\lambda_t$ arrives and has to be processed during the time slot. The jobs can be arbitrarily distributed to the servers. Let $z^\text{max}_j$ denote the maximum job volume that can be processed by one server of type~$j$ during a single time slot.
If a server of type $j$ works with load $z \in [0,z^\text{max}_j]$, it causes cost in the amount of $f_{t,j}(z)$ where $f_{t,j}(z)$ is a convex increasing non-negative function. 
Since $f_{t,j}$ is convex, the cost is minimized if each active server of type $j$ runs with the same load (see Lemma~\ref{lemma:online:func:equally} for a formal proof). Therefore, the \emph{operating} cost for server type~$j$ during time slot $t$ is given by
\begin{equation*}
g_{t,j}(x, z) \coloneqq \begin{cases}
x f_{t,j} \left( \frac{\lambda_t z}{x} \right) & \text{if $x > 0$} \\
\infty & \text{if $x = 0$ and $\lambda_t z > 0$} \\
0 & \text{if $x = 0$ and $\lambda_t z = 0$}
\end{cases}
\end{equation*}
where $x$ is the number of active servers of type $j$ and $z$ is the fraction of the job volume $\lambda_t$ that is assigned to server type $j$. The total operating cost during time slot $t$ is denoted by
\begin{equation} \label{eqn:problem:gt}
g_t(x_1, \dots, x_d) \coloneqq \min_{(z_1, \dots, z_d) \in \mathcal{Z}}  \sum_{j=1}^{d} g_{t,j}(x_j, z_j) 
\end{equation}
where $\mathcal{Z} \coloneqq \{(z_1, \dots, z_d) \in [0,1]^d \mid \sum_{j=1}^{d} z_j = 1 \}$ is the set of all possible job assignments. 


A schedule $X$ is a sequence $\vec{x}_1, \dots, \vec{x}_T$ with $\vec{x}_t = (x_{t,1}, \dots, x_{t,d})$ where each $x_{t,j} \in \{0, 1, \dots, m_j\}$ indicates the number of active servers of type $j$ during time slot $t$. We assume that at the beginning and end of the considered time horizon, all servers are in the inactive state, i.e., $\vec{x}_0 = \vec{x}_{T+1} = (0, \dots, 0)$. 
A schedule is called \emph{feasible}, if there are not more active servers than available and if the maximum load of the active servers is not exceeded, i.e., $x_{t,j} \in \{0, 1, \dots, m_j\}$ and $\sum_{j=1}^{d} x_{t,j} z^\text{max}_j \geq \lambda_t$ holds for all $t \in \{1, \dots, T\}$ and $j \in \{1, \dots, d\}$.
The total cost of a schedule is defined by
\begin{equation} \label{eqn:problem:costx}
C(X) \coloneqq \sum_{t=1}^T \left( g_t(x_{t,1}, \dots, x_{t,d}) + \sum_{j=1}^{d} \beta_j (x_{t,j} - x_{t-1, j})^+ \right)
\end{equation}
where $(x)^+ \coloneqq \max(x, 0)$. 
Note that the switching cost is only paid for powering up. However, this is not a restriction, since all servers are inactive at the beginning and end of the workload. Thus, the cost for powering down can be folded into the cost for powering~up. 

A problem instance is defined by the tuple $\mathcal{I} = (T, d, \vec{m}, \vec{\beta}, F, \Lambda)$ with $\vec{m} = (m_1, \dots, m_d)$, $\vec{\beta} = (\beta_1, \dots, \beta_d)$, $F = (f_{1,1}, \dots, f_{T,d})$ and $\Lambda = (\lambda_1, \dots, \lambda_T)$. The task is to find a schedule with minimal cost. 

In the online version of this problem, the job volumes $\lambda_t$ and the operating cost functions $f_{t,j}$ arrive one-by-one, so $\vec{x}_t$ has to be determined without the knowledge of future jobs $\lambda_{t'}$ and functions $f_{t',j}$ with $t' > t$. 



\textbf{Our contribution.}
We investigate both the online and the offline version of this problem.
In contrast to previous results, we consider the discrete setting where the number of active servers $x_{t,j}$ has to be integral. Thereby, we obtain truly feasible solutions. 

For the online problem, we first examine a simplified version where the operating cost functions $f_{t,j}$ are time-independent (i.e., $f_{t,j} = f_j$ for all $t \in \{1,\dots,T\}$) and present a $(2d + 1)$-competitive deterministic online algorithm (Section~\ref{sec:online}). The basic idea is to calculate an optimal schedule for the problem instance that ends at the current time slot. For each server type, the algorithm ensures that the number of active servers is at least as large as the number of active servers in the optimal schedule.  
A server is powered down if its accumulated idle operating cost $f_j(0)$ exceeds its switching cost $\beta_j$. 
Since the operating cost is time-independent, the runtime of a server can be determined in advance.

In Section~\ref{sec:online:time}, we demonstrate how our algorithm can be modified to handle time-dependent operating cost functions $f_{t,j}$. We achieve a competitive ratio of $2d + 1 + \epsilon$ for any $\epsilon > 0$. The basic idea of the algorithm is unchanged. 
However, in contrast to the previous section, the runtime of a server now depends 
on the time slot when it is powered up, since the idle operating cost $f_{t,j}(0)$ varies over time. Thus, the runtime of a server is not known in advance any more. The analysis results in a competitive ratio of $2d + 1 + c(\mathcal{I})$ where $c(\mathcal{I})$ is a
constant that depends on the switching and operating costs of the problem instance $\mathcal{I}$. By allowing state changes at any time during a time slot and repairing the resulting schedule afterward (such that there are no intermediate state changes any more), we are able to make the constant $c(\mathcal{I})$ arbitrarily small.

In Section~\ref{sec:approx}, we consider the offline version of the problem and present a $(1 + \epsilon)$-approximation algorithm that runs in polynomial time if $d$ is a constant.  
First, we present an optimal algorithm that uses a natural graph representation. The graph is structured in a $(d+1)$-dimensional grid and contains a vertex $v_{t,\vec{x}}$ for each time slot $t \in \{1, \dots, T\}$ and server configuration $\vec{x}$. The vertices are connected with weighted edges that represent the switching and operating costs. By calculating a shortest path, we obtain an optimal schedule. For our approximation algorithm, we only use a small polynomial-sized subset of all vertices depending on the desired approximation factor. Our $(1+\epsilon)$-approximation algorithm runs in $\mathcal{O}\big(T \cdot \epsilon^{-d} \cdot \prod_{j=1}^{d} \log m_j\big)$ time. 
At the end of Section~\ref{sec:approx}, we show that our algorithm still works if the total number of servers varies over time, i.e., $m_j$ is time-dependent.

\textbf{Related work.}
In recent years, energy conservation in data centers has received much attention, see for example \cite{Antoniadis2020,Zhang2018,Albers2017} and references therein. 

Regarding the online version, Lin et al. \cite{LinWierman2013,LinWierman2013extended} analyzed the problem described above for homogeneous data centers where all servers are identical, i.e., $d=1$. The minimum function in equation~\eqref{eqn:problem:gt} disappears, so the operating cost at time slot $t$ is given by $g_t(x) = x f (\lambda_t / x)$, which makes the problem much easier. They presented a 3-competitive online algorithm for the fractional setting where the number of active servers does not need to be integral. This result was improved by Bansal et al. \cite{Bansal2015} who developed a 2-competitive algorithm. 
In our previous paper \cite{AlbersQuedenfeld2018,AlbersQuedenfeld2018extended}, we analyzed the discrete setting for homogeneous data centers. We developed a 3-competitive deterministic and 2-competitive randomized online algorithm and showed that these algorithms are optimal (i.e., there is no algorithm that achieves a better competitive ratio). Furthermore, we proved that 2 is a lower bound for the fractional setting (this result was independently found in \cite{Antoniadis2017}).

The data-center right-sizing problem on \emph{heterogeneous} data centers is related to convex function chasing, also known as smoothed online convex optimization \cite{ChenGoelWierman2018}. 
At each time slot, a convex function $g_t$ arrives and the algorithm has to choose a point $\vec{x}_t \in \mathbb{R}^d$. The cost at time slot $t$ is given by $g_t(\vec{x}_t)$ plus the movement cost $\|\vec{x}_t - \vec{x}_{t-1}\|$ where $\|\cdot \|$ is any metric. Data-center right-sizing in the fractional setting (i.e., the number of active servers can be any real number) is a special case of convex function chasing where $\|\cdot \|$ is a scaled Manhattan metric and the convex functions have the form given in equation~\eqref{eqn:problem:gt}.

Goel and Wierman \cite{GoelWierman2018} developed a $(3 + \mathcal{O}(1/\mu))$-competitive algorithm called Online Balanced Descent (OBD) where the arriving functions are $\mu$-strongly convex. Chen et al. \cite{ChenGoelWierman2018} showed that OBD achieves a competitive ratio of $3 + \mathcal{O}(1/\alpha)$ if the arriving functions are locally $\alpha$-polyhedral.
However, if the operating cost functions $f_{t,j}$ are load-independent, i.e., $f_{t,j}(z) = \text{const}$, then $g_t$ is neither strongly convex nor locally polyhedral, so $\mu = 0$ and $\alpha = 0$. Hence, their results cannot be used for our problem.

Sellke \cite{Sellke2020} developed a $(d+1)$-competitive online algorithm for convex function chasing without any restrictions. A similar result was found by Argue et al. \cite{Argue2020}. The general convex function chasing problem in the \emph{discrete} setting where $g_t$ can be any convex function has (at least) an exponential competitive ratio 
as the following example shows. For all $j \in \{1, \dots, d\}$, let $m_j = 1$ and $\beta_j = 1$, so the feasible server configurations are $\{0,1\}^d$. The arriving functions $g_t$ are infinite for the current position $\vec{x}_{t-1}$ of the online algorithm and zero for all other positions $\{0,1\}^d \setminus \{\vec{x}_{t-1}\}$. The online algorithm always has to change its position to avoid the infinite operating cost (otherwise the online algorithm is not competitive at all). Therefore, after $T \coloneqq 2^d - 1$ time slots, the switching cost of the online algorithm is at least $2^d - 1$. The offline schedule can go directly to a position in $\{0,1\}^d \setminus \bigcup_{t=1}^{T} \{\vec{x}_{t-1}\}$ where no operating cost occurs paying a switching cost of at most $d$. Thus, the competitive ratio for general convex function chasing is at least $\Omega(2^d / d)$.
To gain a competitive ratio with more practical relevance, we focus on operating cost functions described by equation~\eqref{eqn:problem:gt}.

It is an open problem how fractional solutions can be rounded to achieve an integral schedule without significantly increasing the total cost. If the number of active servers is simply rounded up, the total switching cost can get arbitrarily large, for example if the fractional schedule switches permanently between $1$ and $1 + \epsilon$. For homogeneous data centers, a randomized rounding scheme achieving a competitive ratio of 2 was presented in~\cite{AlbersQuedenfeld2018extended}. However, using this method for heterogeneous data centers independently for each server type can lead to an infeasible schedule (e.g., if $\lambda_t = 1$ and $\vec{x}_t = (1/d, \dots, 1/d)$ is rounded down to $(0, \dots, 0)$). 
Thus, Sellke's result does not help us in our analysis of the discrete setting.
Further publications examining the convex body or function chasing problem are \cite{Antoniadis2016,BansalBoehm2018,BubeckSellke2020nested}.

In \cite{AlbersQuedenfeld2021ciac}, we analyzed the discrete setting for heterogeneous data centers where the operating cost does neither depend on the load nor on time , i.e., $f_{t,j}(z) = l_j = \text{const}$. In this case, the total operating cost at time $t$ is given by $g_t(x_1, \dots, x_d) = \sum_{j=1}^{d} l_j x_j$ which is much simpler than the general expression given in equation~\eqref{eqn:problem:gt}. In addition, we assumed that there are no inefficient servers, i.e., a server with a higher switching cost always has a lower operating cost.
We presented a $2d$-competitive algorithm for this special problem. Moreover, we gave a lower bound of $2d$, which also holds for the general problem that we consider in this paper. Thus, our online algorithms presented in Sections~\ref{sec:online} and~\ref{sec:online:time} of this paper are nearly optimal. If the operating cost functions are constant (i.e., $f_{t,j}(z) = \text{const}$), we achieve the optimal competitive ratio of $2d$.

The offline version of the discrete data-center right-sizing problem for \emph{homogeneous} data centers can be solved in polynomial time~\cite{AlbersQuedenfeld2018}. It is an open question whether the 
problem on \textit{heterogeneous} data centers is NP-hard or not. For the special case of load-independent operating costs (i.e., $f_{t,j}(z) = l_j = \text{const}$), a polynomial-time algorithm based on a minimum-cost flow computation was shown in \cite{Albers2017,Albers2019}. However, the flow representation of the problem cannot be generalized for load-dependent operating costs.


Right-sizing of heterogeneous data centers is related to geographical load balancing examined in \cite{LiuLinWierman2011} and \cite{LinLiuWierman2012}. For more works handling related problems, refer to~\cite{Wang2014,Kim2014,Chen2015,BadieiLiWierman2015,GoelChenWierman2017,Zhang2018,LinGoelWierman2020}.

%% file: onlineAlgorithm.tex

In this section we present a $(2d+1)$-competitive deterministic online algorithm for time-independent operating cost functions, i.e., $f_{t,j} = f_j$ for all time slots $t \in [T]$. Roughly, our algorithm works as follows. For each time slot, it calculates an optimal schedule for the job volumes received so far. Servers are powered up such that the number of active servers of each type is at least as large as the number of active servers of the same type in the optimal schedule. A server runs for exactly $\lceil \beta_j / f_j(0) \rceil$ time slots, then it is powered down, regardless of whether or not it was used.
This is similar to the well-known ski rental problem where it is optimal to buy the skis once the total renting cost would exceed the buy price. 

Formally, given the 
problem instance $\mathcal{I} = (T, d, \vec{m}, \vec{\beta}, F, \Lambda)$, the shortened problem instance $\mathcal{I}^t$ is defined by $\mathcal{I}^t \coloneqq (t, d, \vec{m}, \vec{\beta}, F, \Lambda^t)$ with $\Lambda^t = (\lambda_1, \dots, \lambda_t)$. Let $\hat{X}^t$ denote an optimal schedule for this problem instance and let $X^\mathcal{A}$ be the schedule calculated by our algorithm~$\mathcal{A}$. 

Our algorithm works as follows: After calculating $\hat{X}^t$, the algorithm ensures that the number of active servers of each type $j \in [d]$ is greater than or equal to the number of active servers of type $j$ in the last time slot of $\hat{X}^t$. That is, in each time slot $(\hat{x}^t_{t,j} - x^\mathcal{A}_{t-1,j})^+$, servers of type $j$ are powered up such that the inequality $x^\mathcal{A}_{t,j} \geq \hat{x}^t_{t,j}$ is satisfied. A server of type~$j$ is powered down after $\bar{t}_j =  \left\lceil \frac{\beta_j}{f_j(0)} \right\rceil$ time slots. Note that $f_j(0)$ is the operating cost of a server being idle. It does not matter if the server was used or not.

{
\setlength{\textfloatsep}{12pt}
\begin{algorithm}[b] \label{alg:online:func}
	\caption{Algorithm $\mathcal{A}$}
	\begin{algorithmic}[1]
		\State $w_{t,j} \coloneqq 0$ for all $t \in \mathbb{Z}$ and $j \in [d]$
		\For{$t \coloneqq 1$ \textbf{to} $T$}
		\State Calculate $\hat{X}^t$
		\For{$j \coloneqq 1$ \textbf{to} $d$}
		\State $x^\mathcal{A}_{t,j} \coloneqq x^\mathcal{A}_{t,j} - w_{t - \bar{t}_j,j}$
		\If {$x^\mathcal{A}_{t,j} \leq \hat{x}^t_{t,j}$}
		\State $w_{t,j} \coloneqq  \hat{x}^t_{t,j} - x^\mathcal{A}_{t,j}$
		\State $x^\mathcal{A}_{t,j} \coloneqq \hat{x}^t_{t,j}$
		\EndIf
		\EndFor
		\EndFor
	\end{algorithmic}
\end{algorithm}
}

The pseudocode below clarifies how algorithm $\mathcal{A}$ works. 
The schedule $\hat{X}^t$ can be calculated with the optimal offline
 algorithm presented in Section~\ref{sec:approx:opt}.
The variables $w_{t,j}$ store how many servers of type $j$ were powered up at time slot $t$. A visualization of our algorithm is shown in Figure~\ref{fig:online:algo}.


\subsection{Feasibility}
\label{sec:online:func:feasibe}

Before we determine the competitive ratio of our algorithm, we have to show that the calculated schedule is feasible.

\ifincludefigures
\input{image_onlineAlgo.tex}
\fi

\begin{lemma} \label{lemma:online:func:feasible}
	The schedule $X^\mathcal{A}$ is feasible.
\end{lemma}

\begin{proof}
	A schedule is feasible, if (1) $\sum_{j=1}^d x_{t,j} z^\text{max}_j \geq \lambda_t$ and (2) $x_{t,j} \in [m_j]_0$ holds for all $t \in [T]$ and $j \in [d]$.
	 It is always ensured that $x^\mathcal{A}_{t,j} \geq \hat{x}^t_{t,j}$ holds, so condition (1) is satisfied, since $\hat{X}^t$ is a feasible schedule: \begin{equation*}
	 \sum_{j=1}^d x^\mathcal{A}_{t,j} z^\text{max}_j \geq \sum_{j=1}^d \hat{x}^t_{t,j} z^\text{max}_j \geq \lambda_t .
	 \end{equation*}
	 
	 Servers are powered up only in line~8. Since $\hat{X}^t$ is feasible, $x^\mathcal{A}_{t,j} \leq m_j$ is always satisfied. Servers are powered down only in line~5. Each variable $w_{t,j}$ is accessed exactly once, so $x^\mathcal{A}_{t,j}$ never gets negative. Therefore, condition (2) is satisfied. \qedllncs
\end{proof}

\subsection{Competitiveness}
\label{sec:online:func:analysis}

In this section, we will show that algorithm~$\mathcal{A}$ is $(2d+1)$-competitive. 

For our analysis, we split the operating cost into an \emph{idle} and a \emph{load-dependent} part. The \emph{idle} operating cost of an active server of type $j$ for a single time slot is $f_j(0)$, i.e., it does not depend on the load. The \emph{load-dependent} operating cost of all active servers of type~$j$ at time slot $t$ is defined by 
\begin{equation} \label{eqn:online:func:ltj}
	L_{t,j}(X) \coloneqq x_{t,j}  \left( f_j\left(\frac{\lambda_t z_{t,j}}{x_{t,j}}\right) - f_j(0) \right)
\end{equation}
where $z_{t,j}$ are the values $z_j$ that minimize the right term in equation~\eqref{eqn:problem:gt}. Formally, 
\begin{equation*}
(z_{t,1}, \dots, z_{t,d}) \coloneqq \argmin_{(z_1, \dots, z_d) \in \mathcal{Z}} \sum_{j=1}^{d} g_{t,j}(x_{t,j}, z_j).
\end{equation*}
Since $f_j$ is an increasing function, $L_{t,j}(X)$ cannot be negative.

Let $s_{j,1} \leq \dots \leq s_{j,n_j}$ denote the time slots when in $X^\mathcal{A}$ a server of type $j$ is powered up. If $n$ servers of type $j$ are powered up at the same time slot, there are $n$ equal values in the sequence. The time interval $A_{j,i} \coloneqq [s_{j,i} : s_{j,i} + \bar{t}_j -1]$ is called \emph{block} and contains the time slots when the server is in the active state. The switching and idle operating cost of a block $A_{j,i}$ is at most\footnote{If there are two consecutive blocks without a gap between them, there is no switching cost for the second block, so $H_{j,i}$ gives an upper bound for the switching and idle operating cost of $A_{j,i}$.}
\begin{equation} \label{eqn:online:func:dtj}
H_{j,i} \coloneqq \beta_j + \bar{t}_j \cdot f_j(0) .
\end{equation}

For each server type $j \in [d]$ we define special time slots $\tau_{j,1}, \dots, \tau_{j, n'_j}$ that are constructed in reverse time as follows. $\tau_{j, n'_j}$ is defined as the last time slot when a server of type $j$ is powered up in $X^\mathcal{A}$, i.e., $\tau_{j, n'_j} \coloneqq s_{j, n_j}$. Given $\tau_{j,k}$, the previous time slot $\tau_{j,k-1}$ is the last powering up of a server of type $j$ before time slot $\tau_{j,k} - \bar{t}_j$. Formally, for $k < n'_j$, $\tau_{j,k}$ is defined by $\tau_{j,k} \coloneqq \max \{s_{j,i} \mid i \in [n_j],  s_{j,i} \leq \tau_{j,k+1} - \bar{t}_j\}$.
%
Figure~\ref{fig:online:tau} visualizes the definition of~$\tau_{j,k}$.
Since the runtime of a single server is exactly $\bar{t}_j$, this definition ensures that each block~$A_{j,i}$ contains exactly one time slot $\tau_{j,k}$, $k \in [n'_j]$. 

\ifincludefigures
\input{image_tau.tex}
\fi

As already mentioned in the problem description section, the operating costs of all active servers of type~$j$ are minimized if the jobs assigned to type~$j$ are equally distributed to the servers of type~$j$. 
This is formally stated in the lemma below.

\begin{lemma} \label{lemma:online:func:equally}
	Let $f$ be a convex function, $x \in \mathbb{N}$, $\lambda, z \in \mathbb{R}$ and let $\sum_{i=1}^{x} a_i = 1$ with $a_i \geq 0$ for all $i \in [x]$. It holds
	\begin{equation*}
		x f(\lambda z / x) \leq \sum_{i=1}^x f(\lambda z a_i) .
	\end{equation*}
\end{lemma}

We will prove this with Jensen's inequality.

\begin{lemma}[Jensen's inequality] \label{lemma:online:func:jensen}
	Let $f$ be a convex function, $\omega_i \geq 0$ for $i \in [n]$ and $\sum_{i=1}^{n} \omega_i = 1$. It holds
	\begin{equation*}
	f\left(\sum_{i=1}^{n} \omega_i x_i\right) \leq \sum_{i=1}^{n} \omega_i f(x_i).
	\end{equation*}
\end{lemma}

\begin{proof}
	By using Jensen's inequality with $x_i = \lambda z a_i$ and $\omega_i = 1/x$ for $i \in [x]$, we get 
	\begin{equation*}
	f\left(\sum_{i=1}^{x}\lambda z a_i / x\right) \leq \sum_{i=1}^{x} \frac{1}{x} f(\lambda z a_i).
	\end{equation*}
	Multiplying with $x$ and using the fact $\sum_{i=1}^{x} a_i = 1$ gives us
	\begin{equation*}
	x f(\lambda z / x) \leq \sum_{i=1}^{x} f (\lambda z a_i). \qedhere 
	\end{equation*}
\end{proof}

The following lemma states that the load-dependent operating cost of $X^\mathcal{A}$ at time $t$ is less than or equal to that of $\hat{X}^t$.

\begin{lemma}  \label{lemma:online:func:ltj}
	For all $t \in [T]$ and $j \in [d]$, it holds
	\begin{equation*}
		L_{t,j}(X^\mathcal{A}) \leq L_{t,j}(\hat{X}^t).
	\end{equation*}
\end{lemma}

\begin{proof}
	For $i \in [x^\mathcal{A}_{t,j}]$, let 
	\begin{equation*}
		a_i \coloneqq \begin{cases}
		1 / \hat{x}^t_{t,j} & \text{if $i \leq \hat{x}^t_{t,j}$} \\
		0 & \text{otherwise.}
		\end{cases}
	\end{equation*}
	By using the definition of $L_{t,j}$ (equation~\eqref{eqn:online:func:ltj}) and Lemma~\ref{lemma:online:func:equally}, we get
	\begin{align*}
		L_{t,j}(X^\mathcal{A}) &\;=\; x^\mathcal{A}_{t,j} f_j (\lambda_t z_{t,j} / x^\mathcal{A}_{t,j}) - x^\mathcal{A}_{t,j} f_j(0) \\
		&\;\alignstack{ L\ref{lemma:online:func:equally}}{\leq}\; \sum_{i=1}^{x^\mathcal{A}_{t,j}} f_j (\lambda_t z_{t,j} a_i) - x^\mathcal{A}_{t,j} f_j(0) \\
		&\;=\; \sum_{i=1}^{\hat{x}^t_{t,j}} f_j(\lambda_t z_{t,j} / \hat{x}^t_{t,j}) + \sum_{i = \hat{x}^t_{t,j} + 1}^{x^\mathcal{A}_{t,j}} f_j(0) -  x^\mathcal{A}_{t,j} f_j(0) \\
		&\;=\; \hat{x}^t_{t,j} f_j(\lambda_t z_{t,j} / \hat{x}^t_{t,j}) - \hat{x}^t_{t,j} f_j(0) \\
		&\;=\; L_{t,j}(\hat{X}^t).
	\end{align*}
	In the third step, we simply use the definition of $a_i$ and split the sum into two parts. The second sum is  equal to $(x^\mathcal{A}_{t,j} - \hat{x}^t_{t,j}) f_j(0)$. At the end, we use the definition of $L_{t,j}$, again. \qedllncs
\end{proof}

%
By using Lemma~\ref{lemma:online:func:ltj}, we can show that the load-dependent operating cost of $X^\mathcal{A}$ is at most as large as the total cost of the optimal schedule.

\begin{lemma}  \label{lemma:online:func:ltjsum}
	It holds
	\begin{equation*}
		\sum_{t=1}^{T} \sum_{j=1}^{d} L_{t,j}(X^\mathcal{A}) \leq C(\hat{X}^T) .
	\end{equation*}
\end{lemma}

\begin{proof}
	We will prove the inequality 
	\begin{equation*}
	\sum_{t'=1}^{t} \sum_{j=1}^d L_{t', j}(X^\mathcal{A}) \leq C(\hat{X}^t)
	\end{equation*}
	by induction. For $t = 0$, both terms are zero. Assume that
	$\sum_{t'=1}^{t-1} \sum_{j=1}^d L_{t', j}(X^\mathcal{A})  \leq C(\hat{X}^{t-1})$
	holds. Let 
	\begin{equation*}
	C_M(X) \coloneqq \sum_{t \in M} \left( g_t(x_{t,1}, \dots, x_{t,d}) + \sum_{j=1}^{d} \beta_j (x_{t,j} - x_{t-1, j})^+ \right)
	\end{equation*}
	be the switching and operating cost of $X$ for all time slots $t \in M$. Note that the total cost of a schedule is given by $C_{[1:T]}(X) = C(X)$.
	
	Since $\hat{X}^{t-1}$ is an optimal schedule for $\mathcal{I}^{t-1}$, the cost of $\hat{X}^t$ up to the time slot $t-1$ is greater than or equal to $C(\hat{X}^{t-1})$, i.e., $C(\hat{X}^{t-1}) \leq C_{[1:t-1]}(\hat{X}^{t})$. By using this fact as well as the induction hypothesis and Lemma~\ref{lemma:online:func:ltj}, we get\pagebreak[4]
	\begin{align*}
	\sum_{t'=1}^{t} \sum_{j=1}^d L_{t', j}(X^\mathcal{A}) &\;\alignstack{I.H.}{\leq}\; C(\hat{X}^{t-1}) + \sum_{j=1}^{d} L_{t,j}(X^\mathcal{A}) \\
	&\;\alignstack{L\ref{lemma:online:func:ltj}}{\leq}\; C_{[1:t-1]}(\hat{X}^{t}) + \sum_{j=1}^{d} L_{t,j}(\hat{X}^t) \\
	&\;\leq\; C_{[1:t-1]}(\hat{X}^{t}) + C_{\{t\}}(\hat{X}^t) \\
	&\;\leq\; C(\hat{X}^{t}). \qedhere 
	\end{align*}
\end{proof}

%
So far, we found an upper bound for the load-dependent operating cost of $X^\mathcal{A}$. The following lemma is needed to estimate the switching and idle operating cost of $X^\mathcal{A}$ in Lemma~\ref{lemma:online:func:djisum}.

\begin{lemma} \label{lemma:online:func:dji}
	The switching and idle operating cost of the block $A_{j,i}$ is bounded by
	\begin{equation*}
		H_{j,i} \leq 2 \min \{\beta_j + f_j(0), \; \bar{t}_j \cdot f_j(0) \}.
	\end{equation*}
\end{lemma}

\begin{proof}
	By equation~\eqref{eqn:online:func:dtj}, we have $H_{j,i} = \beta_j + \bar{t}_j f_j(0)$. Since 
	\begin{equation*}
	\beta_j \leq \left\lceil\frac{\beta_j}{f_j(0)} \right\rceil \cdot f_j(0) = \bar{t}_j f_j(0),
	\end{equation*} 
	we get $H_{j,i} \leq 2 \bar{t}_j f_j(0) $. 
	Furthermore, due to 
	\begin{equation*}
	\bar{t}_j  f_j(0) \leq \left(\frac{\beta_j}{f_j(0)} + 1 \right) \cdot  f_j(0)  = \beta_j + f_j(0),
	\end{equation*} 
	we get $H_{j,i} = \beta_j + \bar{t}_j f_j(0) \leq 2 \beta_j + f_j(0)$. 
	Therefore, the inequality $H_{j,i} \leq 2 \min \{\beta_j + f_j(0), \; \bar{t}_j \cdot f_j(0)\}$ is satisfied. \qedllncs
\end{proof}

The next lemma shows that the switching and idle operating cost of all servers of type $j$ in $X^\mathcal{A}$ is at most two times the total cost of the optimal schedule.

\begin{lemma}  \label{lemma:online:func:djisum}
	For all $j \in [d]$, it following inequality holds
	\begin{equation} \label{eqn:online:func:djisum}
		\sum_{i = 1}^{n_j}  H_{j,i} \leq 2 \cdot C( \hat{X}^T).
	\end{equation}
\end{lemma}

\begin{proof}
	Let $B_{j,k} \coloneqq \{i \in [n_j] \mid A_{j,i} \ni \tau_{j,k}\}$ with $k \in [n'_j]$ be the indices of the blocks containing the time slot $\tau_{j,k}$ (see Figure~\ref{fig:online:tau}). As already mentioned, each block $A_{j,i}$ contains exactly one time slot $\tau_{j,k}$, so $\bigcup_{k\in[n'_j]} B_{j,k} = [n_j]$ and $B_{j,k} \cap B_{j,k'} = \emptyset$ for $k \not = k'$. Therefore, $\sum_{i = 1}^{n_j}  H_{j,i} = \sum_{k=1}^{n'_j} \sum_{i \in B_{j,k}} H_{j,i}$. 
	
	We will prove equation~\eqref{eqn:online:func:djisum} by induction. To simplify the notation, let $\tau_{j,0} \coloneqq 0$. 
	We will show that 
	\begin{equation} \label{eqn:online:func:djisum:induction}
	\sum_{k=1}^{n} \; \sum_{i \in B_{j,k}}H_{j,i} \leq 2 C(\hat{X}^{\tau_{j,n}})
	\end{equation} 
	holds for all $n \in [n'_j]_0$. For $n = 0$, the inequality is obviously satisfied (the sum is empty and $\hat{X}^0$ is an empty schedule with zero cost). Assume that inequality~\eqref{eqn:online:func:djisum:induction} holds for $n-1$, i.e., $\sum_{k=1}^{n-1} \sum_{i \in B_{j,k}}H_{j,i} \leq 2 C(\hat{X}^{\tau_{j,n-1}})$.

	For $I \subseteq  [T]$, let 
	\begin{equation*}
	C_I(X) \coloneqq \sum_{t \in I} \left( g_t(\vec{x}_t) + \sum_{j=1}^{d} \beta_j (x_{t,j} - x_{t-1,j})^+ \right)
	\end{equation*}
	denote the cost of $X$ during the time interval $I$. 
	We begin from the left-hand side of equation~\eqref{eqn:online:func:djisum:induction}, use our induction hypothesis and get
	\begin{align*}\sum_{k=1}^{n} \; \sum_{i \in B_{j,k}} H_{j,i} \,\alignstack{IH}{\leq}{}\,& 2 C(\hat{X}^{\tau_{j,n-1}}) + \sum_{i \in B_{j,n}} H_{j,i} 
	\\
	\,{}\leq{}\,
	& 
	2C_{[1:\tau_{j,n-1}]}(\hat{X}^{\tau_{j,n}}) + \sum_{i \in B_{j,n}} H_{j,i} \numberthis\label{eqn:online:func:djisum:a}
	\end{align*} 
	The last inequality holds because $\hat{X}^{\tau_{j,n-1}}$ is an optimal schedule for $\mathcal{I}^{\tau_{j,n-1}}$, so $C(\hat{X}^{\tau_{j,n-1}}) \leq C_{[1:\tau_{j,n-1}]}(\hat{X}^{\tau_{j,n}})$. 
	
	By the definition of $\tau_{j,k}$, at time $t \coloneqq \tau_{j,n}$ at least one server of type $j$ is powered up, so 
	\begin{equation}\label{eqn:online:func:dtjsum:mjn}
	\hat{x}^t_{t,j} = x^\mathcal{A}_{t,j} = |B_{j,n}|.
	\end{equation}
	Furthermore, the cost of $\hat{X}^t$ during the time interval $I \coloneqq [\tau_{j,n-1} + 1: \tau_{j,n}]$ is at least 
	\begin{equation} \label{eqn:online:func:djisum:optcostinterval}
	C_{I}(\hat{X}^{\tau_{j,n}}) \geq \hat{x}^t_{t,j} \cdot \min \{\beta_j + f_j(0), f_j(0) \cdot \bar{t}_j\}
	\end{equation}
	because each server of type $j$ that is active at time slot $t$ was powered up during the time interval~$I$ (so there is the switching cost $\beta_j$ as well as the operating cost for at least one time slot) or it was powered up before $I$, so it was active for  $|I| = \tau_{j,n} - \tau_{j,n-1} \geq \bar{t}_j$ time slots. Since $f_j$ is an increasing function, the operating cost is at least $f_j(0)$ for each time slot. 
	
	By using Lemma~\ref{lemma:online:func:dji} and the equations~\eqref{eqn:online:func:dtjsum:mjn} and~\eqref{eqn:online:func:djisum:optcostinterval}, we can transform the term~\eqref{eqn:online:func:djisum:a} to
	\begin{align*}
		 2C_{[1:\tau_{j,n-1}]}(\hat{X}^{\tau_{j,n}}) + \sum_{i \in B_{j,n}} H_{j,i} 
		 \;\alignstack{L\ref{lemma:online:func:dji}}{{}\leq{}}\;& 	2C_{[1:\tau_{j,n-1}]}(\hat{X}^{\tau_{j,n}}) + |B_{j,n}| \cdot 2 \min \{\beta_j + f_j(0), f_j(0) \cdot \bar{t}_j \} \\
		\;\alignstack{\eqref{eqn:online:func:dtjsum:mjn}}{{}\leq{}}\;& 2C_{[1:\tau_{j,n-1}]}(\hat{X}^{\tau_{j,n}}) + 2 \hat{x}^t_{t,j} \min \{\beta_j + f_j(0), f_j(0) \cdot \bar{t}_j \} \\
		\;\alignstack{\eqref{eqn:online:func:djisum:optcostinterval}}{{}\leq{}}\;& 2C_{[1:\tau_{j,n-1}]}(\hat{X}^{\tau_{j,n}}) + 2C_{[\tau_{j,n-1} + 1 : \tau_{j,n}]}(\hat{X}^{\tau_{j,n}}) \\
		\;{}\leq{}\; & 2C(\hat{X}^{\tau_{j,n}}).
	\end{align*}
	Therefore, equation~\eqref{eqn:online:func:djisum:induction} is satisfied for all $n \in [n'_j]_0$. For $n = n'_j$, we get 
	\begin{equation*}
	\sum_{i = 1}^{n_j}  H_{j,i} = \sum_{k=1}^{n'_j} \; \sum_{i \in B_{j,k}} H_{j,i} \leq 2C(\hat{X}^{\tau_j,n'_j}) \leq 2 C(\hat{X}^T). \qedhere 
	\end{equation*}
\end{proof}

Now, we are able to prove the competitive ratio of algorithm~$\mathcal{A}$.

\begin{theorem} \label{theo:online:func}
	Algorithm~$\mathcal{A}$ is $(2d+1)$-competitive.
\end{theorem}

\begin{proof}
	The total cost of $X^\mathcal{A}$ is the switching and idle operating cost given by  $\sum_{j=1}^{d} \sum_{i=1}^{n_j} H_{j,i}$ plus the load-dependent operating cost given by $\sum_{t=1}^{T} \sum_{j=1}^{d} L_{t,j}(X^\mathcal{A})$. By using Lemmas~\ref{lemma:online:func:djisum} and~\ref{lemma:online:func:ltjsum}, we get
	\begin{align*}
	C(X^\mathcal{A}) 
	&\;=\;
	\sum_{j=1}^{d} \sum_{i=1}^{n_j} H_{j,i} + \sum_{t=1}^{T} \sum_{j=1}^{d} L_{t,j}(X^\mathcal{A}) \\
	&\;\alignstack{\substack{\text{\scriptsize L\ref{lemma:online:func:djisum}}\\  \text{\scriptsize L\ref{lemma:online:func:ltjsum}}}}{\leq}\; \sum_{j=1}^{d} 2 \cdot C( \hat{X}^T) + C(\hat{X}^T) \\
	&\;=\;
	 (2d + 1) \cdot C(\hat{X}^T) .
	\end{align*}
	The schedule $\hat{X}^T$ is optimal for the problem instance $\mathcal{I}$, so algorithm~$\mathcal{A}$ is $(2d+1)$-competitive. \qedllncs
\end{proof}


If the operating costs are load independent, i.e., $f_j(z) = l_j = \text{const}$ for all $j \in [d]$, then the load-dependent operating cost $L_{t,j}(X^\mathcal{A})$ is always zero. Thus, the competitive ratio of algorithm~$\mathcal{A}$ is $2d$, so it matches the lower bound given in~\cite{AlbersQuedenfeld2021ciac}. In contrast to the deterministic $2d$-competitive online algorithm presented in~\cite{AlbersQuedenfeld2021ciac}, our algorithm can handle inefficient server types, which were excluded in~\cite{AlbersQuedenfeld2021ciac}.

\begin{corollary}
	If the operating cost functions are load- and time-independent, algorithm~$\mathcal{A}$ achieves an optimal competitive ratio of~$2d$.
\end{corollary}

%% file: image_onlineAlgo.tex
\begin{figure}[t]

	\centering
	\begin{tikzpicture}
		\pgfmathsetmacro{\xBegin}{0}
		\pgfmathsetmacro{\xStep}{0.7} 
		\pgfmathsetmacro{\xLength}{15}
		\pgfmathsetmacro{\tbar}{5}
		\pgfmathsetmacro{\tMax}{\xLength - 1}
		
		\pgfmathsetmacro{\yBeginAlg}{0}
		
		\pgfmathsetmacro{\yLengthAlg}{3} 
		\pgfmathsetmacro{\yLengthOpt}{3}
		\pgfmathsetmacro{\yStep}{0.48}	
		\pgfmathsetmacro{\yBeginArrows}{\yBeginAlg + 3.1*\yStep}
		\pgfmathsetmacro{\yBeginOpt}{\yBeginAlg + \yStep * 3 + 1.3} 
		\pgfmathsetmacro{\yArrowStep}{0.13}
		\pgfmathsetmacro{\yArrowOffset}{2}
		\pgfmathsetmacro{\yLeftRight}{0.04}
		
		\pgfmathsetmacro{\axisExtend}{0.7}
		\pgfmathsetmacro{\tickHalfLength}{0.07}
		
		\pgfmathsetmacro{\numArrows}{4}
		\pgfmathsetmacro{\numArrowsMinusOne}{\numArrows - 1}
		
		\tikzstyle{axis}=[-stealth]
		\tikzstyle{arrow}=[-stealth, thick]
		\tikzstyle{line}=[-, thick]
		
		\def\xOpt{{0,1,1,1,2,2,0,0,3,3,2,2,1,0,0}}
		\def\xAlg{{0,1,1,1,2,2,1,1,3,3,3,3,3,1,0}}
		\def\arrowStart{{1,4,8,9}}

		\fill[fill=red!40!white]
			(\xBegin + 1 * \xStep, \yBeginAlg + 0 * \yStep)
			rectangle
			(\xBegin + 6 * \xStep, \yBeginAlg + 1 * \yStep);
		\fill[fill=blue!40!white]
			(\xBegin + 4 * \xStep, \yBeginAlg + 1 * \yStep)
			rectangle
			(\xBegin + 6 * \xStep, \yBeginAlg + 2 * \yStep);
		\fill[fill=blue!40!white]
			(\xBegin + 6 * \xStep, \yBeginAlg + 0 * \yStep)
			rectangle
			(\xBegin + 9 * \xStep, \yBeginAlg + 1 * \yStep);
		\fill[fill=green!45!white]
			(\xBegin + 8 * \xStep, \yBeginAlg + 1 * \yStep)
			rectangle
			(\xBegin + 9 * \xStep, \yBeginAlg + 3 * \yStep);
		\fill[fill=green!45!white]
			(\xBegin + 9 * \xStep, \yBeginAlg + 0 * \yStep)
			rectangle
			(\xBegin + 13 * \xStep, \yBeginAlg + 2 * \yStep);
		\fill[fill=yellow!50!white]
			(\xBegin + 9 * \xStep, \yBeginAlg + 2 * \yStep)
			rectangle
			(\xBegin + 13 * \xStep, \yBeginAlg + 3 * \yStep);
		\fill[fill=yellow!50!white]
			(\xBegin + 13 * \xStep, \yBeginAlg + 0 * \yStep)
			rectangle
			(\xBegin + 14 * \xStep, \yBeginAlg + 1 * \yStep);
		\fill[fill=red!40!white]
			(\xBegin + 1 * \xStep, \yBeginOpt + 0 * \yStep)
			rectangle
			(\xBegin + 2 * \xStep, \yBeginOpt + 1 * \yStep);
		\fill[fill=blue!40!white]
			(\xBegin + 4 * \xStep, \yBeginOpt + 1 * \yStep)
			rectangle
			(\xBegin + 5 * \xStep, \yBeginOpt + 2 * \yStep);
		\fill[fill=green!45!white]
			(\xBegin + 8 * \xStep, \yBeginOpt + 1 * \yStep)
			rectangle
			(\xBegin + 9 * \xStep, \yBeginOpt + 3 * \yStep);
		\fill[fill=yellow!50!white]
			(\xBegin + 9 * \xStep, \yBeginOpt + 2 * \yStep)
			rectangle
			(\xBegin + 10 * \xStep, \yBeginOpt + 3 * \yStep);
		
		\draw[axis] (\xBegin, \yBeginAlg) 
			to (\xBegin + \xStep * \xLength + \axisExtend * \xStep, \yBeginAlg) node[below] {$t$};
		\draw[axis] (\xBegin, \yBeginAlg) 
			to (\xBegin, \yBeginAlg + \yStep * \yLengthAlg + \axisExtend * \yStep) node[left] {$x^\mathcal{A}_{t,j}$};
		\foreach \x in {0,...,\xLength} {
			\draw[-] 
				(\xBegin + \x * \xStep, \yBeginAlg + \tickHalfLength) -- 
				(\xBegin + \x * \xStep, \yBeginAlg - \tickHalfLength);
		}
		\foreach \y in {0,...,\yLengthAlg} {
			\draw[-] 
			(\xBegin + \tickHalfLength, \yBeginAlg + \y * \yStep) -- 
			(\xBegin - \tickHalfLength, \yBeginAlg + \y * \yStep);
		}
		\foreach \x in {0,...,\tMax} {
			\node[below] at (\xBegin + \x * \xStep + 0.5 * \xStep, \yBeginAlg) {$\x$};
		}
	
		\draw[axis] (\xBegin, \yBeginOpt) 
			to (\xBegin + \xStep * \xLength + \axisExtend * \xStep, \yBeginOpt) node[below] {$t$};
		\draw[axis] (\xBegin, \yBeginOpt) 
			to (\xBegin, \yBeginOpt + \yStep * \yLengthOpt + \axisExtend * \yStep) node[left] {$\hat{x}^t_{t,j}$};
		\foreach \x in {0,...,\xLength} {
			\draw[-] 
			(\xBegin + \x * \xStep, \yBeginOpt + \tickHalfLength) -- 
			(\xBegin + \x * \xStep, \yBeginOpt - \tickHalfLength);
		}
		\foreach \y in {0,...,\yLengthOpt} {
			\draw[-] 
			(\xBegin + \tickHalfLength, \yBeginOpt + \y * \yStep) -- 
			(\xBegin - \tickHalfLength, \yBeginOpt + \y * \yStep);
		}
		\foreach \x in {0,...,\tMax} {
			\node[below] at (\xBegin + \x * \xStep + 0.5 * \xStep, \yBeginOpt) {$\x$};
		}

		\foreach \x in {1,...,\tMax} {
			\pgfmathsetmacro{\ybefore}{\xAlg[\x - 1]}
			\pgfmathsetmacro{\ynow}{\xAlg[\x]}
			\draw[line] 
				(\xBegin + \x * \xStep, \yBeginAlg + \ybefore * \yStep) --
				(\xBegin + \x * \xStep, \yBeginAlg + \ynow * \yStep);
			\draw[line] 
				(\xBegin + \x * \xStep, \yBeginAlg + \ynow * \yStep) --
				(\xBegin + \x * \xStep + \xStep, \yBeginAlg + \ynow * \yStep);
		}

		\foreach \x in {1,...,\tMax} {
			\pgfmathsetmacro{\ybefore}{\xOpt[\x - 1]}
			\pgfmathsetmacro{\ynow}{\xOpt[\x]}
			\draw[line] 
				(\xBegin + \x * \xStep, \yBeginOpt + \ybefore * \yStep) --
				(\xBegin + \x * \xStep, \yBeginOpt + \ynow * \yStep);
			\draw[line] 
				(\xBegin + \x * \xStep, \yBeginOpt + \ynow * \yStep) --
				(\xBegin + \x * \xStep + \xStep, \yBeginOpt + \ynow * \yStep);
		}
	
		\foreach \i in {0,...,\numArrowsMinusOne} {
			\pgfmathsetmacro{\x}{\arrowStart[\i]}
			\pgfmathsetmacro{\yplus}{\yArrowOffset * \yArrowStep + \i * \yArrowStep}
			\draw[arrow] 
			(\xBegin + \x * \xStep + \yLeftRight, 
				\yBeginArrows) --
			(\xBegin + \x * \xStep + \yLeftRight, 
				\yBeginArrows + \yplus) --
			(\xBegin + \x * \xStep + \tbar * \xStep - \yLeftRight,
				\yBeginArrows + \yplus) --
			(\xBegin + \x * \xStep + \tbar * \xStep - \yLeftRight,
				 \yBeginArrows);
		}

	\end{tikzpicture}
	\caption{
		{\normalfont (This figure is colored)} 
		Visualization of algorithm~$\mathcal{A}$ for one specific server type $j$ with $\bar{t}_j = 5$. The upper plot shows $\hat{x}^t_{t,j}$, while the lower plot shows the resulting values $x^\mathcal{A}_{t,j}$. Note that the upper plot is not an optimal schedule, but the last state of each optimal schedule $\hat{X}^1, \hat{X}^2, \dots, \hat{X}^T$. The algorithm ensures that $x^\mathcal{A}_{t,j} \geq \hat{x}^t_{t,j}$ is always satisfied which is visualized by the colors: Each colored square in the upper plot causes a server to be powered up. The runtime of this server is drawn in the same color in the lower plot. Additionally, the arrows indicate the time slot when a server is powered down (e.g., at time slot $1$, a server is powered up, and $\bar{t}_j = 5$ time slots later, it is powered down).  
	}
	\label{fig:online:algo}
\end{figure}
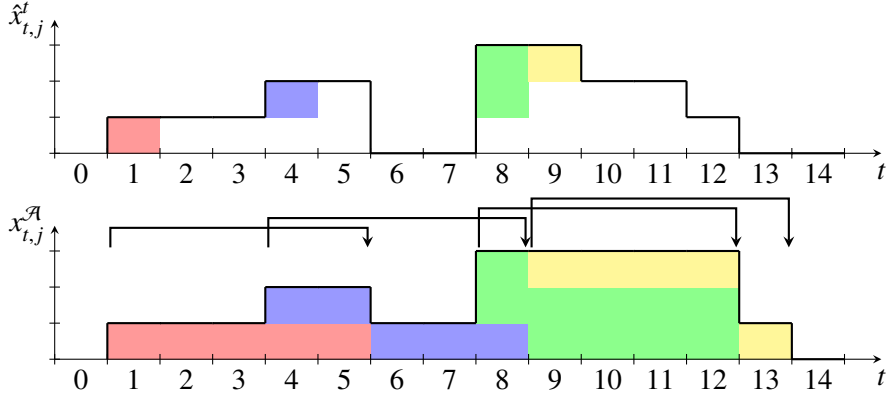

%% file: image_tau.tex
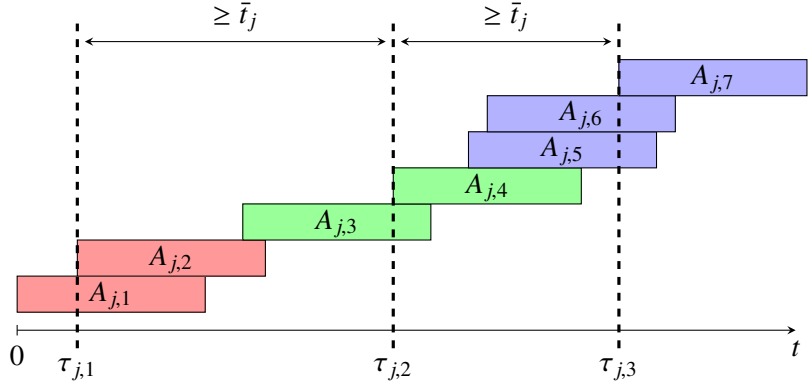
\begin{figure}[t]

	\centering
	\begin{tikzpicture}
		\pgfmathsetmacro{\xBegin}{0}
		\pgfmathsetmacro{\xStep}{0.5} 
		\pgfmathsetmacro{\xLength}{20}
		\pgfmathsetmacro{\tbar}{5}
		\pgfmathsetmacro{\tMax}{\xLength - 1}
		
		\pgfmathsetmacro{\yBegin}{0}
		
		\pgfmathsetmacro{\yStep}{0.48} 
		\pgfmathsetmacro{\yHeightCorrection}{-0.02cm} 
		\pgfmathsetmacro{\numBlocks}{7}
		\pgfmathsetmacro{\numBlocksMinusOne}{\numBlocks - 1}
		\pgfmathsetmacro{\tauLastIndex}{2}
		\pgfmathsetmacro{\tauLastIndexMinusOne}{\tauLastIndex - 1}
		\pgfmathsetmacro{\yAxis}{\yBegin - 0.5 * \yStep}
		\pgfmathsetmacro{\yLow}{\yBegin - \yStep}
		\pgfmathsetmacro{\yHigh}{\yBegin + \numBlocks*\yStep + \yStep}
		
		\pgfmathsetmacro{\axisExtend}{0.7}
		
		\pgfmathsetmacro{\tickHalfLength}{0.07}
		
		\tikzstyle{arrow}=[-stealth]

		\def\activation{{0,1.6,6,10,12,12.5,16}}
		\def\tauarray{{1.6,10,16}}
		\def\colortable{{
				"1.0 0.6 0.6",
				"1.0 0.6 0.6",
				"0.6 1.0 0.6",
				"0.6 1.0 0.6",
				"0.7 0.7 1.0",
				"0.7 0.7 1.0",
				"0.7 0.7 1.0"}}
		
		\foreach \y in {0,...,\numBlocksMinusOne} {
			\pgfmathsetmacro{\yp}{int(\y + 1)}
			\pgfmathsetmacro{\xIdx}{\activation[\y] }
			\pgfmathparse{\colortable[\y]};
			\definecolor{currentColor}{rgb}{\pgfmathresult};
			\draw[fill=currentColor]
				(\xBegin + \xIdx * \xStep, \yBegin + \y * \yStep) 
				rectangle 
				(\xBegin + \xIdx * \xStep + \tbar * \xStep, \yBegin + \y * \yStep + \yStep)
				node[midway, yshift=\yHeightCorrection] {$A_{j,\yp}$};

		}
		
		\draw[arrow] (\xBegin, \yAxis) 
			to (\xBegin + \xStep * \xLength + \axisExtend * \xStep, \yAxis) node[below] {$t$};
		\draw[-] (\xBegin, \yAxis + \tickHalfLength) -- 
			(\xBegin, \yAxis - \tickHalfLength) 
			node[below] {$0$};

		\foreach \y in {0,...,\tauLastIndex} {
			\pgfmathsetmacro{\yp}{int(\y + 1)}
			\pgfmathsetmacro{\xtau}{\tauarray[\y] }
			\draw[dashed, very thick] 
				(\xBegin + \xtau * \xStep, \yLow) -- 
				(\xBegin + \xtau * \xStep, \yHigh);
			\node[below] at
				(\xBegin + \xtau * \xStep, \yLow)
				{$\tau_{j,\yp}$};
		}
	
		\foreach \y in {0,...,\tauLastIndexMinusOne} {
			\pgfmathsetmacro{\yp}{int(\y + 1) }
			\pgfmathsetmacro{\xtau}{\tauarray[\y] }
			\pgfmathsetmacro{\xnext}{\tauarray[\yp] }
			\draw[stealth-stealth] 
				(\xBegin + \xtau * \xStep + 0.1,
				 \yHigh - 0.5 * \yStep) -- 
				(\xBegin + \xnext * \xStep - 0.1,
				 \yHigh - 0.5 * \yStep)
			node[midway,above] {$\geq \bar{t}_j$};
		}

	\end{tikzpicture}
	\caption{
		{\normalfont (This figure is colored)} 
		Example of the blocks $A_{j,i}$ (rectangles) and the corresponding special time slots $\tau_{j,k}$ (dashed vertical lines) for one specific server type $j$. The distance between two consecutive special time slots is always greater than or equal to $\bar{t}_j$. The index block sets $B_{j,k}$ defined in the proof of Lemma~\ref{lemma:online:func:djisum} are $B_{j,1} = \{1,2\}$ (marked in red), $B_{j,2} = \{3,4\}$ (green), $B_{j,3} = \{5,6,7\}$ (blue). 
	}
	\label{fig:online:tau}
\end{figure}

%% file: onlineAlgorithmTime.tex
In this section, we present a modified version of algorithm~$\mathcal{A}$ that is able to handle time-dependent operating cost functions $f_{t,j}$ and achieves a competitive ratio of $2d + 1 + \epsilon$ for any $\epsilon > 0$. The proof is divided into two parts.
First, as an intermediate result we introduce algorithm~$\mathcal{B}$ that is $\left(2d + 1 + \sum_{j = 1}^{d}\max_{t \in [T]} \frac{f_{t,j}(0)}{\beta_j}\right)$-competitive. 
Then, in Subsection~\ref{sec:online:time:epsilon} 
we show how the given problem instance $\mathcal{I}$ can be modified to make the constant $c(\mathcal{I}) \coloneqq \sum_{j = 1}^{d}\max_{t \in [T]} \frac{f_{t,j}(0)}{\beta_j}$ arbitrarily small. Finally, the resulting schedule is adapted to the original problem instance without increasing its cost. 

\subsection{Obtaining a competitive ratio of $2d + 1 + c(\mathcal{I})$}
\label{sec:online:time:ci}

%

To handle time-dependent operating cost functions, algorithm~$\mathcal{A}$ has to be modified, as the idle operating cost $f_{t,j}(0)$ is no longer constant over time. 
Similar to algorithm~$\mathcal{A}$, in algorithm~$\mathcal{B}$ a server is powered down when its accumulated idle operating cost $f_{t,j}(0)$ exceeds its switching cost. Formally, 
let $l_{t,j} \coloneqq f_{t,j}(0)$ be the idle operating cost of server type $j$ during time slot $t$ and let 
\begin{equation*}
	\bar{t}_{t,j} \coloneqq \max \left\{\bar{t} \in [T-t] \mid \sum_{u = t+1}^{t+\bar{t} } l_{u,j} \leq \beta_j\right\}
\end{equation*} 
be the maximal number of time slots such that the sum of the idle operating costs beginning from time slot $t+1$ is smaller than or equal to $\beta_j$. 
%
A server that is powered up at time slot $t$ runs for $\bar{t}_{t,j}$ \emph{further} time slots, i.e., it is powered down at time slot $t + \bar{t}_{t,j}$. This definition differs from $\bar{t}_j$ in algorithm~$\mathcal{A}$ where a server is powered down at $t + \bar{t}_j - 1$. Note that the idle operating cost at time slot $t$ does not influence the runtime of a server. 
The power-up policy of algorithm~$\mathcal{B}$ is the same as in algorithm~$\mathcal{A}$, i.e., it is always ensured that the number of active servers of type $j$ is at least as large as the corresponding number in an optimal schedule for the problem instance that ends at the current time slot. Formally, $x^\mathcal{B}_{t,j} \geq \hat{x}^t_{t,j}$ holds for all $t \in [T]$ and $j \in [d]$.

In contrast to algorithm~$\mathcal{A}$, the runtime of a server is not known when it is powered up, because the future operating cost functions did not arrive yet, so $\bar{t}_{t,j}$ cannot be calculated at this time. However, the runtime is known at the time slot when the server must be powered down, so $\mathcal{B}$ is a valid online algorithm. 
The pseudocode below clarifies how algorithm~$\mathcal{B}$ works. Note that only lines~5 and~6 change in comparison to algorithm~$\mathcal{A}$. The set $W_t$ defined in line~5 contains all time slots $u$ with $u + \bar{t}_{u,j} + 1= t$. Servers that were powered up at time slot $u$ are shut down at time slot $t$. Figure~\ref{fig:online:time} visualizes the definition of $\bar{t}_{t,j}$ and $W_t$ and shows an example of how algorithm~$\mathcal{B}$ operates. 

\ifincludefigures
\input{image_onlineAlgoTime.tex}
\fi

\begin{algorithm}[H] \label{alg:online:time}
	\caption{Algorithm $\mathcal{B}$}
	\begin{algorithmic}[1]
		\State $w_{t,j} \coloneqq 0$ for all $t \in \mathbb{Z}$ and $j \in [d]$
		\For{$t \coloneqq 1$ \textbf{to} $T$}
		\State Calculate $\hat{X}^t$
		\For{$j \coloneqq 1$ \textbf{to} $d$}
		\State $W_t \coloneqq \big\{u \in [t-1] | \sum_{v = u+1}^{t-1} l_{v,j} \leq \beta_j < \sum_{v = u+1}^{t} l_{v,j}\big\}$
		\State $x^\mathcal{B}_{t,j} \coloneqq x^\mathcal{B}_{t,j} - \sum_{u \in W_t} w_{u,j}$  
		\If {$x^\mathcal{B}_{t,j} \leq \hat{x}^t_{t,j}$}
		\State $w_{t,j} \coloneqq  \hat{x}^t_{t,j} - x^\mathcal{B}_{t,j}$
		\State $x^\mathcal{B}_{t,j} \coloneqq \hat{x}^t_{t,j}$
		\EndIf
		\EndFor
		\EndFor
	\end{algorithmic}
\end{algorithm}

Before we analyze the competitive ratio of algorithm~$\mathcal{B}$, we have to prove that the calculated schedule $X^\mathcal{B}$ is feasible.

\begin{lemma}\label{lemma:online:time:feasible}
	The schedule $X^\mathcal{B}$ is feasible.
\end{lemma}

\begin{proof}
	A schedule is feasible, if (1) $\sum_{j=1}^d x_{t,j} z^\text{max}_j \linebreak[0] \geq \lambda_t$ and (2) $x_{t,j} \in [m_j]_0$ holds for all $t \in [T]$ and $j \in [d]$.
	Analogously to algorithm~$\mathcal{A}$, it is always ensured that $x^\mathcal{B}_{t,j} \geq \hat{x}^t_{t,j}$. Since $\hat{X}^t$ is a feasible schedule, we get
	\begin{equation*}
	\sum_{j=1}^d x^\mathcal{B}_{t,j} z^\text{max}_j \geq \sum_{j=1}^d \hat{x}^t_{t,j} z^\text{max}_j \geq \lambda_t ,
	\end{equation*}
	
	so property (1) is satisfied.
	
	Servers are powered up only in line~9. Since $\hat{X}^t$ is feasible, $x^\mathcal{B}_{t,j} \leq m_j$ is always satisfied. Servers are powered down only in line~6. To ensure that $x^\mathcal{B}_{t,j}$ never gets negative, we have to show that each variable $w_{t,j}$ is accessed at most one time. This is equivalent to $W_t \cap W_{t'} = \emptyset$ for all $t < t'$. 
	
	For $t < t'$, it holds 
	\begin{align*}
	W_t \cap W_{t'} &= \left\{u \in [t-1] | \sum_{v = u+1}^{t-1} l_{v,j} \leq \beta_j < \sum_{v = u+1}^{t} l_{v,j}\right\} 
	\cap \left\{u \in  [t'-1] | \sum_{v = u+1}^{t'-1} l_{v,j} \leq \beta_j < \sum_{v = u+1}^{t'} l_{v,j}\right\} \\
	&\subseteq \left\{u \in [t'-1] | \sum_{v = u+1}^{t-1} l_{v,j} \leq \beta_j < \sum_{v = u+1}^{t} l_{v,j}\right\} 
	\cap \left\{u \in  [t'-1] | \sum_{v = u+1}^{t'-1} l_{v,j} \leq \beta_j < \sum_{v = u+1}^{t'} l_{v,j}\right\} \\
	&= \left\{u \in [t'-1] | \left(\sum_{v = u+1}^{t-1} l_{v,j} \leq \beta_j < \sum_{v = u+1}^{t} l_{v,j} \right) 
	\land \left(\sum_{v = u+1}^{t'-1} l_{v,j} \leq \beta_j < \sum_{v = u+1}^{t'} l_{v,j}\right)\right\} \\
	&\subseteq \left\{u \in [t'-1] | \beta_j < \sum_{v = u+1}^{t} l_{v,j} \leq \sum_{v = u+1}^{t'-1} l_{v,j} \leq \beta_j \right\} \\
	&= \emptyset
	\end{align*}
	In the first step, we just insert the definition of $W_t$. Then, the first set is expanded by the elements $u \in [t : t'-1]$. In the third step, both sets are connected. Afterward, we use the fact that $t \leq t' - 1$. The resulting set must be empty, because the condition $\beta_j < \dots \leq \beta_j$ is never satisfied. Therefore, property (2) holds. \qedllncs
\end{proof}


 The analysis of the competitive ratio of algorithm~$\mathcal{B}$ is quite similar to that of algorithm~$\mathcal{A}$. Let 
\begin{equation*}
	L_{t,j}(X) \coloneqq x_{t,j}\left( f_{t,j}\left( \frac{\lambda_t z_{t,j}}{x_{t,j}}\right) - l_{t,j} \right)
\end{equation*}
denote the load-dependent operating cost of $X$. Lemmas~\ref{lemma:online:func:equally} and~\ref{lemma:online:func:ltj} still hold, since in their proofs we can simply replace $f_j$ with $f_{t,j}$. Lemma~\ref{lemma:online:func:ltjsum} directly follows from Lemma~\ref{lemma:online:func:ltj}, so it also remains applicable.

The schedule $X^\mathcal{B}$ is divided into blocks $A_{j,i} \coloneqq [s_{j,i} : s_{j,i} + \bar{t}_{t,j}]$ (the definition of $s_{j,i}$ remains the same). The switching and idle operating cost of a block $A_{j,i}$ is at most
\begin{equation} \label{eqn:online:time:hji}
	H_{j,i} \coloneqq \beta_j + \sum_{u = s}^{s + \bar{t}_{s,j}} l_{u,j}
\end{equation}
with $s = s_{j,i}$. The special time slots $\tau_{j,k}$ are defined in the same way as in the previous section. Formally, they are given by $\tau_{j,n'_j} \coloneqq s_{j,n_j}$ and $\tau_{j,k} \coloneqq \max \{s_{j,i} \mid i \in [n_j], s_{j,i} + \bar{t}_{s_{j,i}, j} < \tau_{j,k+1} \}$ for $1 \leq k < n'_j$ as well as $\tau_{j,0} \coloneqq 0$. The definition of the index sets $B_{j,k} = \{i \in [n_j] \mid A_{j,i} \ni \tau_{j,k}\}$ do not change. The following lemma replaces Lemma~\ref{lemma:online:func:dji} and gives an upper bound for $H_{j,i}$. 

\begin{lemma}\label{lemma:online:time:dji}
	The switching and idle operating cost of $A_{j,i}$ is at most
	\begin{equation*}
	H_{j,i} \leq 2 \beta_j + \max_{t \in [T]} l_{t,j}.
	\end{equation*}
\end{lemma}

\begin{proof}
	Let $s \coloneqq s_{j,i}$. By the definition of $\bar{t}_{s,j}$, we know that $\sum_{u = s + 1}^{s + \bar{t}_{s,j}} \leq \beta_j$. We use this inequality in equation~\ref{eqn:online:time:hji} and get $H_{j,i} = \beta_j + \sum_{u = s}^{s + \bar{t}_{s,j}} l_{u,j} \leq  2\beta_j + l_{s,j} \leq 2 \beta_j + \max_{t \in [T]} l_{t,j}$. \qedllncs
\end{proof}

The next lemma replaces Lemma~\ref{lemma:online:func:djisum} and shows that the switching and idle operating costs caused by server type $j$ are at most $2 + \max_{t \in [T]} { l_{t,j}}/{\beta_j}$ times larger than the total cost of an optimal schedule.

\begin{lemma}\label{lemma:online:time:djisum}
	For all $j \in [d]$, it holds
	\begin{equation} \label{eqn:online:time:djisum}
	\sum_{i = 1}^{n_j}  H_{j,i} \leq \left(2 + \max_{t \in [T]} \frac{ l_{t,j}}{\beta_j}\right) \cdot C( \hat{X}^T).
	\end{equation}
\end{lemma}

\begin{proof}
	This proof works very similar to the proof of Lemma~\ref{lemma:online:func:djisum}. 
	Each block $A_{j,i}$ contains exactly one special time slot $\tau_{j,k}$, so $\sum_{i = 1}^{n_j}  H_{j,i} = \sum_{k=1}^{n'_j} \sum_{i \in B_{j,k}} H_{j,i}$. 	
	We will show by induction that 
	\begin{equation} \label{eqn:online:time:djisum:induction}
	\sum_{k=1}^{n} \; \sum_{i \in B_{j,k}}H_{j,i} \leq (2 + c_j) \cdot C(\hat{X}^{\tau_{j,n}})
	\end{equation} 
	with $c_j \coloneqq \max_{t \in [T]} { l_{t,j}} / {\beta_j}$ holds for all $n \in [n'_j]_0$. For $n = 0$, the inequality is obviously satisfied (the sum is empty and $\hat{X}^0$ is an empty schedule with zero cost). 
	Assume that inequality~\eqref{eqn:online:time:djisum:induction} holds for $n-1$, i.e., $\sum_{k=1}^{n-1} \sum_{i \in B_{j,k}}H_{j,i} \leq (2 + c_j) \cdot C(\hat{X}^{\tau_{j,n-1}})$. %
	\newcommand{\tauLast}{u}%
	\newcommand{\tauNow}{t}%
	To simplify the notation, let $\tauLast \coloneqq \tau_{j,n-1}$ and $\tauNow \coloneqq \tau_{j,n}$. 
	We begin from the left hand side of equation~\eqref{eqn:online:time:djisum:induction}, use our induction hypothesis and get
	\begin{align*}\sum_{k=1}^{n} \; \sum_{i \in B_{j,k}} H_{j,i} &\stackrel{IH}{\leq} (2 + c_j) \cdot  C(\hat{X}^{\tauLast}) + \sum_{i \in B_{j,n}} H_{j,i} \\
	&\leq (2 + c_j) \cdot  C_{[1:\tauLast]}(\hat{X}^{\tauNow}) + \sum_{i \in B_{j,n}} H_{j,i} \numberthis\label{eqn:online:time:djisum:a}
	\end{align*}
	The last inequality holds because $\hat{X}^{\tauLast}$ is an optimal schedule for $\mathcal{I}^{\tauLast}$, so $C(\hat{X}^{\tauLast}) = C_{[1:\tauLast]}(\hat{X}^{\tauLast}) \leq C_{[1:\tauLast]}(\hat{X}^{\tauNow})$. 
	
	By the definition of the special time slots $\tau_{j,k}$, at time $t$ at least one server of type $j$ is powered up, so 
	\begin{equation}\label{eqn:online:time:dtjsum:mjn}
	\hat{x}^t_{t,j} = x^\mathcal{B}_{t,j} = |B_{j,n}|.
	\end{equation}
	Furthermore, the cost of $\hat{X}^t$ during the time interval $I \coloneqq [\tauLast + 1: \tauNow]$ is at least 
	\begin{equation*} 
	C_{I}(\hat{X}^{\tauNow}) \geq \hat{x}^t_{t,j} \cdot \min \left\{\beta_j + l_{\tauNow,j}, \sum_{t' = \tauLast+1}^{\tauNow} l_{t',j} \right\}
	\end{equation*}
	because each server of type $j$ that is active at time slot $t$ was powered up during the time interval~$I$ (so there is the switching cost $\beta_j$ as well as the operating cost for time slot $t$) or it was powered up before $I$, so it was active during the time interval $I$. Since $f_{t,j}$ are increasing functions, the operating costs are at least $f_{t,j}(0) = l_{t,j}$ for each time slot.

	Let $\bar{t} \coloneqq \bar{t}_{u,j}$. By the definition of $\tau_{j,n-1}$ we have $\tauLast + \bar{t} < \tauNow$, so $\tauLast + \bar{t} + 1 \leq \tauNow$ holds. The definition of $\bar{t}_{\tauLast,j}$ implies that $\beta_j < \sum_{t'=\tauLast + 1}^{\tauLast + \bar{t} + 1} l_{t',j}$. Therefore, we get $\beta_j < \sum_{t'=\tauLast + 1}^{\tauNow} l_{t',j}$, so the cost of $\hat{X}^t$ during $I$ is at least
	\begin{equation}\label{eqn:online:time:djisum:optcostinterval}
	C_I\hat{X}^{\tauNow}) \geq \hat{x}^t_{t,j} \cdot \min \left\{\beta_j + l_{\tauNow,j}, \sum_{t' = \tauLast+1}^{\tauNow} l_{t',j} \right\} \geq \hat{x}^t_{t,j} \cdot \beta_j  .
	\end{equation}
	
	By using Lemma~\ref{lemma:online:time:dji} and the equations~\eqref{eqn:online:time:dtjsum:mjn} and~\eqref{eqn:online:time:djisum:optcostinterval}, we can transform the term~\eqref{eqn:online:time:djisum:a} to
	\begin{align*}
	(2 + c_j) \cdot C_{[1:\tauLast]}(\hat{X}^{\tauNow}) + \sum_{i \in B_{j,n}} H_{j,i} 
	\stackrel{L\ref{lemma:online:time:dji}}{{}\leq{}}& 	(2 + c_j) \cdot C_{[1:\tauLast]}(\hat{X}^{\tauNow}) + |B_{j,n}| \cdot \left(2 \beta_j + \max_{t \in [T]} l_{t,j} \right) \\
	\stackrel{\eqref{eqn:online:time:dtjsum:mjn}}{{}\leq{}}& (2 + c_j) \cdot C_{[1:\tauLast]}(\hat{X}^{\tauNow}) + \hat{x}^t_{t,j} \cdot \left(2 \beta_j + \max_{t \in [T]} l_{t,j} \right)  \\
	\stackrel{\eqref{eqn:online:time:djisum:optcostinterval}}{{}\leq{}}& (2 + c_j) \cdot C_{[1:\tauLast]}(\hat{X}^{\tauNow}) + C_{[\tauLast + 1 : \tauNow]}(\hat{X}^{\tauNow}) \cdot \left(2 + \max_{t \in [T]} \frac{l_{t,j}}{\beta_j} \right)\\
	{}={} & (2 + c_j) \cdot C(\hat{X}^{\tauNow}).
	\end{align*}
	The last equality just uses the definition of $c_j$. 
	Therefore, equation~\eqref{eqn:online:time:djisum:induction} is satisfied for all $n \in [n'_j]_0$. For $n = n'_j$, we get 
	\begin{equation*}
	\sum_{i = 1}^{n_j}  H_{j,i} = \sum_{k=1}^{n'_j} \; \sum_{i \in B_{j,k}} H_{j,i} \leq (2 + c_j) \cdot C(\hat{X}^{\tau_j,n'_j}) \leq (2 + c_j) \cdot  C(\hat{X}^T). \qedhere 
	\end{equation*}
\end{proof}

Now, we are able to prove the competitive ratio of algorithm~$\mathcal{B}$.

\begin{theorem} \label{theo:online:time}
	Algorithm~$\mathcal{B}$ is $(2d+1+ c(\mathcal{I}))$-competitive with $c(\mathcal{I}) = \sum_{j = 1}^{d}\max_{t \in [T]} \frac{l_{t,j}}{\beta_j}$. 
\end{theorem}

\begin{proof}
	The total cost of $X^\mathcal{B}$ is the switching and idle operating costs given by  $\sum_{j=1}^{d} \sum_{i=1}^{n_j} H_{j,i}$ plus the load-dependent operating cost given by $\sum_{t=1}^{T} \sum_{j=1}^{d} L_{t,j}(X^\mathcal{B})$. By using Lemma~\ref{lemma:online:time:djisum} and~\ref{lemma:online:func:ltjsum}, we get
	\begin{align*}
	C(X^\mathcal{B}) 
	&\stackrel{\phantom{L\ref{lemma:online:time:djisum},L\ref{lemma:online:func:ltjsum}}}{=}
	\sum_{j=1}^{d} \sum_{i=1}^{n_j} H_{j,i} + \sum_{t=1}^{T} \sum_{j=1}^{d} L_{t,j}(X^\mathcal{B}) \\
	&\stackrel{L\ref{lemma:online:time:djisum},L\ref{lemma:online:func:ltjsum}}{\leq} 
	\sum_{j=1}^{d} \left(2 + \max_{t \in [T]}\frac{ l_{t,j}}{\beta_j}\right) \cdot C( \hat{X}^T) + C(\hat{X}^T) \\
	&\stackrel{\phantom{L\ref{lemma:online:time:djisum},L\ref{lemma:online:func:ltjsum}}}{=}
	\left(2d + 1 + \sum_{j = 1}^{d} \max_{t \in [T]} \frac{ l_{t,j}}{\beta_j}\right) \cdot C(\hat{X}^T).
	\end{align*}
	The schedule $\hat{X}^T$ is optimal for the problem instance $\mathcal{I}$, so algorithm~$\mathcal{B}$ is $(2d+1+ c(\mathcal{I}))$-competitive. \qedllncs
\end{proof}


\subsection{Improving the competitive ratio to $2d+1+\epsilon$}
\label{sec:online:time:epsilon}

In the following, we show how the competitive ratio can be reduced to $2d+1+\epsilon$ for any $\epsilon > 0$.
Given the original problem instance $\mathcal{I} = (T,d, \vec{m}, \vec{\beta}, F, \Lambda)$, we consider the modified problem instance $\tilde{\mathcal{I}} = (\tilde{T}, d, \vec{m}, \vec{\beta}, \tilde{F}, \tilde{\Lambda})$ where each time slot $t$ of the original problem instance is divided into $\tilde{n}_t$ equal sub time slots. The values $\tilde{n}_t \in \mathbb{N}$ are defined later. The total number of time slots is given by $\tilde{T} \coloneqq \sum_{t=1}^{T} \tilde{n}_t$. In the following, time slots in the original problem instance $\mathcal{I}$ are denoted by $t$, whereas time slots in the modified problem instance $\tilde{\mathcal{I}}$ are denoted by $u$. 
Let $U(t)$ be the set of time slots in the modified problem instance $\tilde{\mathcal{I}}$ that corresponds to the time slot $t \in [T]$ in the original problem instance $\mathcal{I}$. Formally, $U(t) \coloneqq [u + 1 : u + \tilde{n}_t]$ with $u = \sum_{t'=1}^{t-1} \tilde{n}_{t'}$. Furthermore, we define $U^{-1}(u)$ with $u \in [\tilde{T}]$ to be the time slot $t \in [T]$ such that $u \in U(t)$. 
%
The operating cost functions of $\tilde{\mathcal{I}}$ are defined as 
\begin{equation*}
\tilde{f}_{u,j}(z) \coloneqq \frac{1}{\tilde{n}_{t}} f_{t,j}(z)
\end{equation*}
with $t = U^{-1}(u)$ for all $u \in [\tilde{T}]$ and $j \in [d]$, so the operating cost during time slot $t$ is divided into $\tilde{n}_t$ equal parts. The idle operating cost is denoted by $\tilde{l}_{u,j} \coloneqq \tilde{f}_{u,j}(0)$ for $u \in [\tilde{T}]$ and $j \in [d]$. The job volumes do not change, so $\lambda_{u} \coloneqq \lambda_{U^{-1}(u)}$ for all $u \in [\tilde{T}]$. 
In other words, $\tilde{\mathcal{I}}$ matches the problem instance $\mathcal{I}$ where intermediate state changes are allowed. 

Let $n \in \mathbb{N}$. We set $\tilde{n}_t = n \cdot \max_{j\in [d]} \frac{l_{t,j}}{\beta_j}$ and apply algorithm~$\mathcal{B}$ on the corresponding problem instance $\tilde{\mathcal{I}}$. Therefore, we get
\begin{equation}\label{eqn:online:time:epsilon:cf}
c(\tilde{\mathcal{I}}) = \sum_{j = 1}^{d}\max_{u \in \tilde{T}} \frac{\tilde{l}_{u,j}}{\beta_j} = \sum_{j = 1}^{d}\max_{t \in T} \frac{l_{t,j}}{\tilde{n}_t \beta_j} \leq \sum_{j = 1}^{d}\max_{t \in \tilde{T}} \frac{1}{n} = \frac{d}{n} .
\end{equation}
In the second step, we use that $\tilde{l}_{u,j} = l_{t,j} / \tilde{n}_t$ with $t = U^{-1}(u)$. The inequality holds because $\tilde{n}_t \geq n \cdot \frac{l_{t,j}}{\beta_j}$ for all $j \in [d]$. 
To achieve a competitive ratio of $2d + 1 + \epsilon$, we set $n = d/\epsilon$. For $n \rightarrow \infty$, the competitive ratio converges to $2d + 1$. 

We still have to show how the resulting $(2d + 1 + \epsilon)$-competitive schedule for the modified problem instance $\tilde{\mathcal{I}}$ can be transformed into a feasible schedule for the original problem instance $\mathcal{I}$. Let $X^\mathcal{B}$ be the schedule created by $\mathcal{B}$ and let $X^\mathcal{C}$ be the final schedule for~$\mathcal{I}$. 
For each original time slot $t \in [T]$, let $\vec{x}^\mathcal{C}_t \coloneqq {\vec{x}}^\mathcal{B}_{\mu(t)}$ with $\mu(t) \coloneqq \argmin_{u \in U(t)} \tilde{g}_{u}({\vec{x}}^\mathcal{B}_{u})$ be the server configuration that minimizes the operating cost during the time interval $U(t)$.  

The pseudocode below shows how the schedule $X^\mathcal{C}$ is calculated. For each arriving operating cost function $f_{t,j}$, the next $\tilde{n}_t$ time slot of $\tilde{\mathcal{I}}$ are created and passed to algorithm~$\mathcal{B}$. Afterward, the next server configuration $\vec{x}^\mathcal{C}_t$ is determined. The whole schedule $X^\mathcal{B}$ cannot be calculated at once, because the state $\vec{x}^\mathcal{C}_t$ must be fixed before the next function $f_{t+1,j}$ can be processed. 

\begin{algorithm}[H] \label{alg:online:time:epsilon}
	\caption{Algorithm $\mathcal{C}$}
	\begin{algorithmic}[1]
		\State Initialize algorithm~$\mathcal{B}$
		\For{$t \coloneqq 1$ \textbf{to} $T$}
			\State Create the next $\tilde{n}_t$ time slots of the modified problem\newline\hspace*{25pt} instance $\tilde{\mathcal{I}}$ with $\tilde{n}_t \coloneqq  {d}/{\epsilon} \cdot \max_{j\in [d]}  {l_{t,j}}/{\beta_j}$
			\State Execute $\tilde{n}_t$ time slots in algorithm~$\mathcal{B}$
			\State $\vec{x}^\mathcal{C}_t \coloneqq {\vec{x}}^\mathcal{B}_{\mu(t)}$ with $\mu(t) \coloneqq \argmin_{u \in U(t)} \tilde{g}_{u}({\vec{x}}^\mathcal{B}_{u})$
		\EndFor
	\end{algorithmic}
\end{algorithm}

The following lemma shows that this procedure does not increase the cost of the schedule.

\begin{lemma} \label{lemma:online:time:epsilon:costcb}
	The total cost of $X^\mathcal{C}$ regarding the problem instance $\mathcal{I}$ is smaller than or equal to the total cost of $X^\mathcal{B}$ regarding the modified problem instance $\tilde{\mathcal{I}}$.
\end{lemma}

\begin{proof}
	Let $C^\mathcal{J}_\text{op}(X)$ be the operating cost of the schedule $X$ regarding the problem instance $\mathcal{J} \in \{\mathcal{I}, \tilde{\mathcal{I}}\}$ and let $C^\mathcal{J}_\text{sw}(X)$ denote its switching cost. 
	

	First, we will compare the operating cost of both schedules.
	The operating cost of $X^\mathcal{B}$ is given by
	\begin{align*}
		C^{\tilde{\mathcal{I}}}_\text{op}(X^\mathcal{B}) 
		= \sum_{u = 1}^{\tilde{T}} \tilde{g}_{u}(\vec{x}^\mathcal{B}_{u}) 
		= \sum_{t = 1}^{T} \sum_{u \in U(t)} \tilde{g}_{u}(\vec{x}^\mathcal{B}_{u}) 
		&\geq  \sum_{t = 1}^{T} \tilde{n}_t \cdot \min_{u \in U(t)} \tilde{g}_{u}(\vec{x}^\mathcal{B}_{u}).
	\end{align*}
	For the last inequality, we estimate each summand by the minimum of all summands. By using the definition of $\vec{x}^\mathcal{C}_t$, we get
	\begin{equation*}
		\sum_{t = 1}^{T} \tilde{n}_t \cdot \min_{u \in U(t)} \tilde{g}_{u}(\vec{x}^\mathcal{B}_{u}) =  \sum_{t = 1}^{T} \tilde{n}_t \cdot \min_{u \in U(t)} \tilde{g}_{u}(\vec{x}^\mathcal{C}_t).
	\end{equation*}
	The definition of $\tilde{f}_{u,j}$ implies that $\tilde{g}_{u}(\vec{x}) = \frac{1}{\tilde{n}_t} g_t(\vec{x}) $ with $t = U^{-1}(u)$, so
	\begin{align*}
		\sum_{t = 1}^{T} \tilde{n}_t \cdot \min_{u \in U(t)} \tilde{g}_{u}(\vec{x}^\mathcal{C}_t) 
		&=  \sum_{t = 1}^{T} g_{t}(\vec{x}^\mathcal{C}_t) 
		= C^\mathcal{I}_\text{op}(X^\mathcal{C}).
	\end{align*}
	Altogether we have shown that $C^{\tilde{\mathcal{I}}}_\text{op}(X^\mathcal{B})  \geq C^\mathcal{I}_\text{op}(X^\mathcal{C})$.
	
	Next, we will compare the switching cost. To simplify the notation, let 
	$S(\vec{x}, \vec{x}') \coloneqq \sum_{j=1}^{d} \beta_j (x'_{j} - x_{j})^+$ 
	be the switching cost from the state $\vec{x}$ to $\vec{x}'$. The total switching cost of $X^\mathcal{B}$ is given by 
	\begin{equation*}
		C^{\tilde{\mathcal{I}}}_\text{sw}(X^\mathcal{B}) 
		= \sum_{u=1}^{\tilde{T}} S(\vec{x}^\mathcal{B}_{u-1}, \vec{x}^\mathcal{B}_u)
		= \sum_{t=1}^{T+1} \sum_{u = \mu(t-1) + 1}^{\mu(t)} S(\vec{x}^\mathcal{B}_{u-1}, \vec{x}^\mathcal{B}_{u}) 
	\end{equation*}
	with $\mu(0) \coloneqq 0$ and $\mu(T+1) \coloneqq \tilde{T} + 1$. In the last step, the interval $[\tilde{T}]$ is partitioned into the sub-intervals $[1 : \mu(1)], [\mu(1) + 1: \mu(2)], \dots, [\mu(T) + 1: \tilde{T} + 1]$ (note that the switching cost from time slot $\tilde{T}$ to $\tilde{T} + 1$ is always 0, since $\vec{x}^\mathcal{B}_{\tilde{T} + 1} = 0$ by definition). The switching cost during each interval is at least as large as the switching cost for jumping directly to the last state of the interval. Therefore, 
	\begin{equation*}
		\sum_{t=1}^{T+1} \sum_{u = \mu(t-1) + 1}^{\mu(t)} S(\vec{x}^\mathcal{B}_{u-1}, \vec{x}^\mathcal{B}_{u}) 
		\leq \sum_{t=1}^{T+1} S(\vec{x}^\mathcal{B}_{\mu(t-1)}, \vec{x}^\mathcal{B}_{\mu(t)}).
	\end{equation*}
	By using the definition of $\vec{x}^\mathcal{C}_t$, we get 
	\begin{equation*}
		\sum_{t=1}^{T+1} S(\vec{x}^\mathcal{B}_{\mu(t-1)}, \vec{x}^\mathcal{B}_{\mu(t)})
		= \sum_{t=1}^{T+1} S(\vec{x}^\mathcal{C}_{t-1}, \vec{x}^\mathcal{C}_{t})
		= C^{\mathcal{I}}_\text{sw}(X^\mathcal{C}),
	\end{equation*}
	so $C^{\tilde{\mathcal{I}}}_\text{sw}(X^\mathcal{B}) \geq C^{\mathcal{I}}_\text{sw}(X^\mathcal{C})$. \qedllncs
\end{proof}

%
%

Now, we can prove that algorithm~$\mathcal{C}$ is $(2d + 1 + \epsilon)$-competitive. 

\begin{theorem} \label{theo:online:time:espilon}
	For any $\epsilon > 0$, there is a $(2d +1 + \epsilon)$-competitive algorithm for the data-center right-sizing problem with heterogeneous servers and time-dependent operating cost functions.
\end{theorem}

\begin{proof}
	Let $X^\ast_{\mathcal{J}}$ be an optimal schedule for the problem instance $\mathcal{J} \in \{\mathcal{I}, \tilde{\mathcal{I}}\}$ and let $C^\mathcal{J}(X)$ denote the total cost of $X$ with respect to $\mathcal{J}$. We have to show that $C^\mathcal{I}(X^\mathcal{C}) \leq (2d + 1 + \epsilon) \cdot C^\mathcal{I}(X^\ast_\mathcal{I})$. 
	By using Lemma~\ref{lemma:online:time:epsilon:costcb}, Theorem~\ref{theo:online:time} and the competitive ratio of algorithm~$\mathcal{B}$ given by equation~\eqref{eqn:online:time:epsilon:cf}, we get
	\begin{align*}
		C^\mathcal{I}(X^\mathcal{C}) &\;\;\alignstack{L\ref{lemma:online:time:epsilon:costcb}}{\leq}\;\; C^{\tilde{\mathcal{I}}}(X^\mathcal{B}) \\ 
		&\;\;\alignstack{T\ref{theo:online:time}}{\leq}\;\; \big(2d + 1 + c(\tilde{\mathcal{I}})\big) \cdot C^{\tilde{\mathcal{I}}}(X^\ast_{\tilde{\mathcal{I}}})  \\
		&\;\;\alignstack{\eqref{eqn:online:time:epsilon:cf}}{\leq}\;\; (2d + 1 + \epsilon) \cdot C^{\tilde{\mathcal{I}}}(X^\ast_{\tilde{\mathcal{I}}}) \leq C^{\mathcal{I}}(X^\ast_{\mathcal{I}}).
	\end{align*}
	The last inequality holds because each feasible schedule $X$ for the problem instance $\mathcal{I}$ can be converted into a feasible schedule $\tilde{X}$ for the modified problem instance $\tilde{\mathcal{I}}$ without increasing the cost. 
	Formally, the definition $\tilde{\vec{x}}_u \coloneqq \vec{x}_{U^{-1}(u)}$ for all $u \in [\tilde{T}]$ implies $C^{\tilde{\mathcal{I}}}(\tilde{X}) = C^\mathcal{I}(X)$. 
	Therefore, an optimal schedule for $\mathcal{I}$ cannot have a lower cost than an optimal schedule for~$\tilde{\mathcal{I}}$. \qedllncs
\end{proof}

%% file: image_onlineAlgoTime.tex
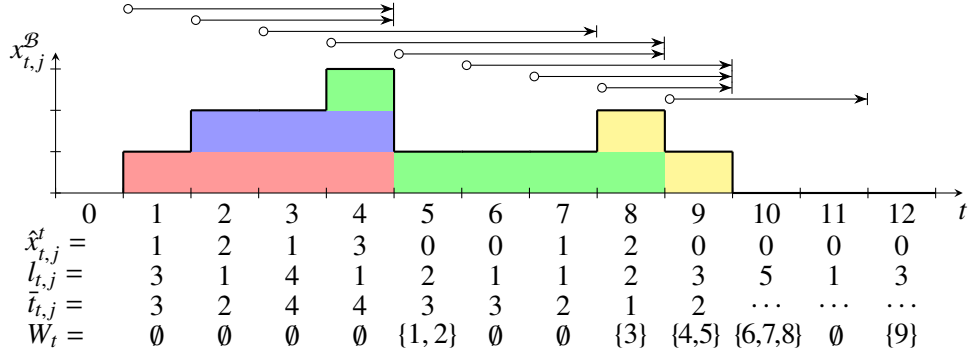
\begin{figure}[tb]
	
	\centering
	\begin{tikzpicture}
		\pgfmathsetmacro{\xBegin}{0}
		\pgfmathsetmacro{\xStep}{0.9} 
		\pgfmathsetmacro{\xLength}{13}
		\pgfmathsetmacro{\tMax}{\xLength - 1}
		
		\pgfmathsetmacro{\yBeginAlg}{0}
		\pgfmathsetmacro{\yStep}{0.55} 
		
		\pgfmathsetmacro{\yopt}{-0.7}
		\pgfmathsetmacro{\yNumberDist}{-0.4} 
		\pgfmathsetmacro{\yltj}{\yopt + \yNumberDist}
		\pgfmathsetmacro{\yttj}{\yltj + \yNumberDist}
		\pgfmathsetmacro{\ywt}{\yttj + \yNumberDist}
		
		\pgfmathsetmacro{\yarrow}{2.6} 
		\pgfmathsetmacro{\yarrowstep}{-0.15}
		
		\pgfmathsetmacro{\yLengthAlg}{3}
		\pgfmathsetmacro{\yLengthOpt}{3}
		\pgfmathsetmacro{\yBeginOpt}{\yBeginAlg + \yLengthAlg * \yStep + \yStep + 10*\yarrowstep}

		\pgfmathsetmacro{\axisExtend}{0.4}
		\pgfmathsetmacro{\tickHalfLength}{0.07}
		
		\tikzstyle{axis}=[-stealth]
		\tikzstyle{arrow}=[{Circle[open]}-{Stealth}{Bar}]
		\tikzstyle{line}=[-, thick]
		\tikzstyle{numbernode}=[]
		
		\def\xOpt{{0,1,2,1,3,0,0,1,2,0,0,0,0}}
		\def\xAlg{{0,1,2,2,3,1,1,1,2,1,0,0,0}}
		\def\ltj{{0,3,1,4,1,2,1,1,2,3,5,1,3}}
		
		\node[numbernode] at (\xBegin + \xStep * 0, \yopt) {$\hat{x}^t_{t,j} = $};
		\foreach \x in {1,...,\tMax} {
			\pgfmathsetmacro{\value}{\xOpt[\x]}
			\node[numbernode] at (\xBegin + \x * \xStep + 0.5*\xStep, \yopt) {\value};
		}
		
		\node[numbernode] at (\xBegin + \xStep * 0, \yltj) {$l_{t,j} = $};
		\foreach \x in {1,...,\tMax} {
			\pgfmathsetmacro{\ltjx}{\ltj[\x]}
			\node[numbernode] at (\xBegin + \x * \xStep + 0.5*\xStep, \yltj) {\ltjx};
		}
		
		\node[numbernode] at (\xBegin + \xStep * 0, \yttj) {$\bar{t}_{t,j} = $};
		\node[numbernode] at (\xBegin + \xStep * 1.5, \yttj) {$3$};
		\node[numbernode] at (\xBegin + \xStep * 2.5, \yttj) {$2$};
		\node[numbernode] at (\xBegin + \xStep * 3.5, \yttj) {$4$};
		\node[numbernode] at (\xBegin + \xStep * 4.5, \yttj) {$4$};
		\node[numbernode] at (\xBegin + \xStep * 5.5, \yttj) {$3$};
		\node[numbernode] at (\xBegin + \xStep * 6.5, \yttj) {$3$};
		\node[numbernode] at (\xBegin + \xStep * 7.5, \yttj) {$2$};
		\node[numbernode] at (\xBegin + \xStep * 8.5, \yttj) {$1$};
		\node[numbernode] at (\xBegin + \xStep * 9.5, \yttj) {$2$};
		\node[numbernode] at (\xBegin + \xStep * 10.5, \yttj) {$\ldots$};
		\node[numbernode] at (\xBegin + \xStep * 11.5, \yttj) {$\ldots$};
		\node[numbernode] at (\xBegin + \xStep * 12.5, \yttj) {$\ldots$};
		
		\newcommand{\spcl}{\hskip-0.2pt}
		\newcommand{\spck}{\hskip-1.9pt}
		\newcommand{\spc}{\hskip-0.7pt}
		\node[numbernode] at (\xBegin + \xStep * 0, \ywt) {$W_t = $};
		\node[numbernode] at (\xBegin + \xStep * 1.5, \ywt) {$\emptyset$};
		\node[numbernode] at (\xBegin + \xStep * 2.5, \ywt) {$\emptyset$};
		\node[numbernode] at (\xBegin + \xStep * 3.5, \ywt) {$\emptyset$};
		\node[numbernode] at (\xBegin + \xStep * 4.5, \ywt) {$\emptyset$};
		\node[numbernode] at (\xBegin + \xStep * 5.5, \ywt) {$\{1,2\}$};
		\node[numbernode] at (\xBegin + \xStep * 6.5, \ywt) {$\emptyset$};
		\node[numbernode] at (\xBegin + \xStep * 7.5, \ywt) {$\emptyset$};
		\node[numbernode] at (\xBegin + \xStep * 8.5, \ywt) {$\{3\}$};
		\node[numbernode] at (\xBegin + \xStep * 9.5, \ywt) {$\{\spc4{\hskip-0.2pt},\spck5\spc\}\,$};
		\node[numbernode] at (\xBegin + \xStep * 10.5, \ywt) {$\,\{\spc6{\hskip-0.2pt},\spck7{\hskip-1pt},\spck8\spc\}$};
		\node[numbernode] at (\xBegin + \xStep * 11.5, \ywt) {$\,\emptyset$};
		\node[numbernode] at (\xBegin + \xStep * 12.5, \ywt) {$\{9\}$};
		
		\draw[arrow] (\xBegin + \xStep * 1.0, \yarrow + \yarrowstep * 1) -- (\xBegin + \xStep * 5, \yarrow + \yarrowstep * 1);
		\draw[arrow] (\xBegin + \xStep * 2.0, \yarrow + \yarrowstep * 2) -- (\xBegin + \xStep * 5, \yarrow + \yarrowstep * 2);
		\draw[arrow] (\xBegin + \xStep * 3.0, \yarrow + \yarrowstep * 3) -- (\xBegin + \xStep * 8, \yarrow + \yarrowstep * 3);
		\draw[arrow] (\xBegin + \xStep * 4.0, \yarrow + \yarrowstep * 4) -- (\xBegin + \xStep * 9, \yarrow + \yarrowstep * 4);
		\draw[arrow] (\xBegin + \xStep * 5.0, \yarrow + \yarrowstep * 5) -- (\xBegin + \xStep * 9, \yarrow + \yarrowstep * 5);
		\draw[arrow] (\xBegin + \xStep * 6.0, \yarrow + \yarrowstep * 6) -- (\xBegin + \xStep * 10, \yarrow + \yarrowstep * 6);
		\draw[arrow] (\xBegin + \xStep * 7.0, \yarrow + \yarrowstep * 7) -- (\xBegin + \xStep * 10, \yarrow + \yarrowstep * 7);
		\draw[arrow] (\xBegin + \xStep * 8.0, \yarrow + \yarrowstep * 8) -- (\xBegin + \xStep * 10, \yarrow + \yarrowstep * 8);
		\draw[arrow] (\xBegin + \xStep * 9.0, \yarrow + \yarrowstep * 9) -- (\xBegin + \xStep * 12, \yarrow + \yarrowstep * 9);
		
		\fill[fill=red!40!white]
			(\xBegin + 1 * \xStep, \yBeginAlg + 0 * \yStep)
			rectangle
			(\xBegin + 5 * \xStep, \yBeginAlg + 1 * \yStep);
		\fill[fill=blue!40!white]
			(\xBegin + 2 * \xStep, \yBeginAlg + 1 * \yStep)
			rectangle
			(\xBegin + 5 * \xStep, \yBeginAlg + 2 * \yStep);
		\fill[fill=green!45!white]
			(\xBegin + 4 * \xStep, \yBeginAlg + 2 * \yStep)
			rectangle
			(\xBegin + 5 * \xStep, \yBeginAlg + 3 * \yStep);
		\fill[fill=green!45!white]
			(\xBegin + 5 * \xStep, \yBeginAlg + 0 * \yStep)
			rectangle
			(\xBegin + 9 * \xStep, \yBeginAlg + 1 * \yStep);
		\fill[fill=yellow!50!white]
			(\xBegin + 8 * \xStep, \yBeginAlg + 1 * \yStep)
			rectangle
			(\xBegin + 9 * \xStep, \yBeginAlg + 2 * \yStep);
		\fill[fill=yellow!50!white]
			(\xBegin + 9 * \xStep, \yBeginAlg + 0 * \yStep)
			rectangle
			(\xBegin + 10 * \xStep, \yBeginAlg + 1 * \yStep);
			

		\draw[axis] (\xBegin, \yBeginAlg) 
			to (\xBegin + \xStep * \xLength + \axisExtend * \xStep, \yBeginAlg) node[below] {$t$};
		\draw[axis] (\xBegin, \yBeginAlg) 
			to (\xBegin, \yBeginAlg + \yStep * \yLengthAlg + \axisExtend * \yStep) node[left] {$x^\mathcal{B}_{t,j}$};
		\foreach \x in {0,...,\xLength} {
			\draw[-] 
				(\xBegin + \x * \xStep, \yBeginAlg + \tickHalfLength) -- 
				(\xBegin + \x * \xStep, \yBeginAlg - \tickHalfLength);
		}
		\foreach \y in {0,...,\yLengthAlg} {
			\draw[-] 
			(\xBegin + \tickHalfLength, \yBeginAlg + \y * \yStep) -- 
			(\xBegin - \tickHalfLength, \yBeginAlg + \y * \yStep);
		}
		\foreach \x in {0,...,\tMax} {
			\node[below] at (\xBegin + \x * \xStep + 0.5 * \xStep, \yBeginAlg) {$\x$};
		}
		
	
		\foreach \x in {1,...,\tMax} {
			\pgfmathsetmacro{\ybefore}{\xAlg[\x - 1]}
			\pgfmathsetmacro{\ynow}{\xAlg[\x]}
			\draw[line] 
				(\xBegin + \x * \xStep, \yBeginAlg + \ybefore * \yStep) --
				(\xBegin + \x * \xStep, \yBeginAlg + \ynow * \yStep);
			\draw[line] 
				(\xBegin + \x * \xStep, \yBeginAlg + \ynow * \yStep) --
				(\xBegin + \x * \xStep + \xStep, \yBeginAlg + \ynow * \yStep);
		}
		

	\end{tikzpicture}
	\caption{
		{\normalfont (This figure is colored)} 
		Visualization of algorithm $\mathcal{B}$ for one specific server type $j$ with $\beta_j = 6$. The plot shows the number of active servers $x^\mathcal{B}_{t,j}$ of algorithm $\mathcal{B}$. The colors indicate the running time of each server. 
		The values $\hat{x}^t_{t,j}$ (that are needed to determine when a server has to be powered up) and the idle operating costs $l_{t,j}$ as well as the resulting values of $\bar{t}_{t,j}$ and $W_t$ are shown below the plot. 
		The running time $\bar{t}_{t,j}$ of a server that is powered up a time slot $t$ is indicated by the arrows, e.g., a server that is powered up at time slot $t = 2$ runs for $\bar{t}_{2,j} = 2$ \emph{additional} time slots, so it is powered down at the end of time slot $t + \bar{t}_{t,j} = 4$. The values $\bar{t}_{t,j}$ are the maximal number of time slots \emph{after} $t$ such that the idle operating costs do not exceed $\beta_j$, e.g., $\bar{t}_{2,j} = 2$, because $l_{3,j} + l_{4,j} = 4+1 = 5 \leq \beta_j = 6$, but $l_{3,j} + l_{4,j} + l_{5,j} = 7 > \beta_j$. At time slot $t$, the servers that were powered up at time $u \in W_t$ are shut down, e.g., $W_5 = \{1,2\}$, so both the red and the blue server are powered down at time slot $5$. For $t \geq 10$, the values of $\bar{t}_{t,j}$ are not known yet, because they depend on $l_{13,j}$.
	}
	\label{fig:online:time}
\end{figure}

%% file: approxAlgorithm.tex
In this section, we consider the offline version of the data-center right-sizing problem and present a $(1+ \epsilon)$-approximation algorithm that runs in $\mathcal{O}\big(T \cdot \epsilon^{-d} \cdot \prod_{j=1}^{d} \log m_j\big)$ time, which is polynomial if $d$ is a constant. It is based on an optimal, graph-based algorithm that is presented in the following subsection. Afterward, in Section~\ref{sec:approx:eps}, we show how the optimal algorithm can be modified to obtain a $(1+\epsilon)$-approximation in polynomial time.

To simplify the following calculations we introduce some notations. Let $M_j \coloneqq [m_j]_0$ and $\mathcal{M} \coloneqq \varprod_{j=1}^d M_j$ be the set of all possible server configurations. The operating cost is denoted by $C_\text{op}(X) \coloneqq \sum_{t=1}^{T} g_t(\vec{x}_t)$ and the switching cost is denoted by $C_\text{sw}(X) \coloneqq \sum_{t=1}^T \sum_{j=1}^{d} \beta_j (x_{t,j} - x_{t-1, j})^+$. 

\subsection{Optimal offline algorithm}
\label{sec:approx:opt}

An optimal schedule can be found by converting the problem instance $\mathcal{I}$ to a graph and finding the shortest path. 

The graph $G(\mathcal{I})$ (or simply denoted by $G$) contains $2T \cdot \prod_{j=1}^d (m_j + 1)$ vertices arranged in a $(d+1)$-dimensional grid (where the first dimension has $2T$ layers). 
For each time slot $t \in [T]$ and each server configuration $\vec{x} = (x_1, \dots, x_d) \in \mathcal{M}$, there are two vertices in the graph denoted by $v_{t, \vec{x}}^\uparrow$ and $v_{t,\vec{x}}^\downarrow$. 
There is an edge $e_{t,\vec{x}}^\text{op}$ from $v_{t, \vec{x}}^\uparrow$ to $v_{t,\vec{x}}^\downarrow$ with weight $g_t(\vec{x})$ representing the operating cost during time slot~$t$.
For each $j \in [d]$ and for each
\begin{equation*}
\vec{x} = (x_1, \dots, x_d) \in M_1 \times \dots \times (M_j \setminus \{m_j\}) \times \dots \times M_d
\end{equation*}
(note that $x_j = m_j$ is excluded), let $\vec{x}' \coloneqq (x_1, \dots, x_j + 1, \dots, x_d)$. 
There is an edge $e^\uparrow_{t, \vec{x}, j}$ from $v_{t, \vec{x}}^\uparrow$ to $v_{t, \vec{x}'}^\uparrow$ with weight $\beta_j$ (a server of type $j$ is powered up) and another edge $e^\downarrow_{t, \vec{x}', j}$from $v_{t, \vec{x}'}^\downarrow$ to $v_{t, \vec{x}}^\downarrow$ with weight $0$ (a server of type $j$ is powered down). Furthermore, for each $t \in [T-1]$ we need an edge $e^\rightarrow_{t, \vec{x}}$ from $v_{t,\vec{x}}^\downarrow$ to $v_{t+1, \vec{x}}^\uparrow$ with weight $0$ to switch to the next time slot.

Let $\vec{0} \coloneqq (0, \dots, 0) \in \mathcal{M}$. Each schedule $X$ for the problem instance $\mathcal{I}$ can be represented by a path $P_X$ between $v_{1,\vec{0}}^\uparrow$ and $v_{T,\vec{0}}^\downarrow$. For each $t \in [T]$, the path uses the edge $e^\text{op}_{t, \vec{x}_t}$. The vertices $v^\downarrow_{t,\vec{x}_t}$ and $v^\uparrow_{t+1,\vec{x}_{t+1}}$ (for $t \in [T-1]$) are connected by an arbitrary shortest path between them. The same is done for the start and the end point. Note that the sum of the weights of the path's edges is equal to the cost of the schedule. 

On the other hand, a given path $P$ between $v_{1,\vec{0}}^\uparrow$ and $v_{T,\vec{0}}^\downarrow$ represents a schedule $X^P$. 
If the path uses the edge $e^\text{op}_{t,\vec{x}}$, then the corresponding schedule uses the server configuration $\vec{x}$ during time slot~$t$. If $P$ does not use a shortest path between $v^\downarrow_{t,\vec{x}_t}$ and $v^\uparrow_{t+1,\vec{x}_{t+1}}$ (for $t \in [T-1]$), then the sum of the weights of the path's edges are greater than the cost of the corresponding schedule. However, by replacing the path's vertices between $v^\downarrow_{t,\vec{x}_t}$ and $v^\uparrow_{t+1,\vec{x}_{t+1}}$ for all $t \in [T-1]$ by a shortest sub-path, both values are equal.

A shortest path between $v_{1,\vec{0}}^\uparrow$ and $v_{T,\vec{0}}^\downarrow$ corresponds to an optimal schedule. Owing to the graph structure, a shortest path can be calculated with dynamic programming in $\mathcal{O}(T \cdot \prod_{j=1}^{d} m_j)$ time. Note that this runtime is not polynomial (even if $d$ is a constant), because the encoding length of the problem instance is $\mathcal{O}(T + \sum_{j=1}^{d} \log m_j)$. The graph structure and the relation between a shortest path and an optimal schedule are visualized in Figure~\ref{fig:approx:graph}.

\ifincludefigures
\input{image_graph.tex}
\fi

\subsection{$(1+\epsilon)$-approximation}
\label{sec:approx:eps}

In this section, we develop a $(1+\epsilon)$-approximation which has a polynomial runtime, if $d$ and $\epsilon$ are constants. 
The basic idea is to reduce the number of possible values for $x_{t,j}$, that is, we will calculate an optimal solution where the number of active servers can only take specific values. Broadly speaking, the number of active servers are powers of a constant $\gamma > 1$.
For example, we will see that using the values $x_{t,j} \in \{0, 1, 2, 4, 8, \dots, m_j\}$ (i.e., each power of two up to $m_j$ as well as $m_j$ and $0$) would result in a $3$-approximation. 
The set of values that will be used for the number of active servers of type $j$ is
\begin{align*}
M^{\gamma}_j \coloneqq{}& \{0, m_j\} \;\cup\; \{\lfloor \gamma^k \rfloor \in M_j \mid k \in \mathbb{N} \} \;\cup\; \{\lceil \gamma^k \rceil \in M_j \mid k \in \mathbb{N} \} \\
={}& \{0, 1, \lfloor \gamma^1 \rfloor, \lceil \gamma^1 \rceil, \lfloor \gamma^2 \rfloor, \lceil \gamma^2 \rceil, \dots, m_j\}.
\end{align*}
Using both the rounded down and rounded up values of $\gamma^k$ ensures that the ratio between two consecutive values is not larger than $\gamma$. Note that $|M^{\gamma}_j| \in \mathcal{O}(\log_\gamma m_j)$. 
Furthermore, we define $\mathcal{M}^{\gamma} \coloneqq \varprod_{j=1}^d M^{\gamma}_j$ as the set of server configurations that will be used in our algorithm. For a given value $x_j < m_j$, let $N_j(x_j)$ be the next greater value of $x_j$ in $M^{\gamma}_j$, i.e., $N_j(x_j) \coloneqq \min \{x \in M^{\gamma}_j \mid x > x_j\}$. 

The reduced graph $G^{\gamma}$ contains the vertices $v_{t,\vec{x}}^s$ with $s \in \{\uparrow, \downarrow\}$, $t \in [T]$ and $\vec{x} \in \mathcal{M}^{\gamma}$. 
Similar to $G$ there is an edge from $v_{t,\vec{x}}^\uparrow$ to $v_{t,\vec{x}}^\downarrow$ with weight $g_t(\vec{x})$ (for all $t\in [T]$ and $\vec{x} \in \mathcal{M}^{\gamma}$) and an edge from $v_{t,\vec{x}}^\downarrow$ to $v_{t+1,\vec{x}}^\uparrow$ with cost 0 (for all $t \in [T-1]$ and $\vec{x} \in \mathcal{M}^{\gamma}$). For each $j \in [d]$ and for each 
\begin{equation*}
\vec{x}  = (x_1, \dots, x_d) \in M^{\gamma}_1 \times \dots \times (M^{\gamma}_j \setminus \{m_j\}) \times \dots \times M^{\gamma}_d,
\end{equation*}
let $\vec{x}' \coloneqq (x_1, \dots, N_j(x_j), \dots, x_d)$. There is an edge from $v_{t,\vec{x}}^\uparrow$ to $v_{t,\vec{x}'}^\uparrow$ with weight $\beta_j (N_j(x_j) - x_j)$ and an edge from $v_{t, \vec{x}'}^\downarrow$ to $v_{t, \vec{x}}^\downarrow$ with weight $0$.

\begin{theorem} \label{theo:approx:eps:theorem}
	Let $P^{\gamma}$ be a shortest path in $G^{\gamma}$ and $X^{\gamma}$ the corresponding schedule. Let $X^\ast$ be an optimal schedule for the original problem instance. Then, the inequality
	\begin{equation}
	C(X^{\gamma}) \leq (2\gamma - 1) \cdot C(X^\ast)
	\end{equation}
	is satisfied, i.e., $X^{\gamma}$ is a $(2\gamma - 1)$-approximation.
\end{theorem}

To prove this theorem, we construct a path $P'$ in $G^{\gamma}$ with the corresponding schedule $X'$ that is not necessarily a shortest path, however, it will satisfy the inequality $C(X') \leq (2\gamma - 1) \cdot C(X^\ast)$. The cost of $X^{\gamma}$ can only be smaller, because the corresponding path $P^{\gamma}$ is a shortest path in $G^{\gamma}$, so if $X'$ is a $(2\gamma - 1)$-approximation, then $X^{\gamma}$ is a $(2\gamma - 1)$-approximation too.

	Given the optimal solution $X^\ast$ the states of $X'$ are defined by
	\begin{equation} \label{eqn:approx:eps:rule}
	x'_{t,j} = \left\{
	\begin{array}{ll} 
	x_\text{min} &\ \text{if}\ x'_{t-1,j} \leq x^\ast_{t,j} \\
	x'_{t-1,j} &\ \text{if}\  x^\ast_{t,j} < x'_{t-1,j} \leq (2\gamma - 1) \cdot x^\ast_{t,j} \\
	x_\text{max} &\ \text{if}\ (2\gamma - 1) \cdot  x^\ast_{t,j} < x'_{t-1,j} 
	\end{array}\right.
	\end{equation}
	with $x_\text{min} = \min\{x \in M^{\gamma}_j \mid x \geq x^\ast_{t,j}\}$ and $x_\text{max} = \max\{x \in M^{\gamma}_j \mid x \leq (2\gamma - 1) \cdot x^\ast_{t,j}\}$
	for all $t \in [T]$ and $j \in [d]$ (with $\vec{x}'_0 = \vec{0}$). Note that the invariant 
	\begin{equation}  \label{eqn:approx:eps:invariant}
	x^\ast_{t,j} \leq x'_{t,j} \leq (2\gamma - 1) \cdot x^\ast_{t,j}
	\end{equation}
	is always satisfied. The construction of $X'$ is visualized in Figure~\ref{fig:approx:xprime}. 
	
	\ifincludefigures
	\input{image_approxXprime.tex}
	\fi
	
	For the proof of Theorem~\ref{theo:approx:eps:theorem}, we will first show that the operating cost of $X'$ is a $(2\gamma - 1)$-approximation for the operating cost of $X^\ast$. For this, we need the following two technical lemmas.
	
	\begin{lemma} \label{lemma:approx:eps:min:a}
		Let $a_1, \dots, a_d \in [1, \infty[$ and let $h_1, \dots, h_d$ be arbitrary non-negative functions. It holds 
		\begin{equation}
		\min_{(z_1, \dots, z_d) \in \mathcal{Z}}  \sum_{j=1}^{d} a_j h_j(z_j) \geq \min_{(z_1, \dots, z_d) \in \mathcal{Z}}  \sum_{j=1}^{d} h_j(z_j).
		\end{equation}
	\end{lemma}
	
	\begin{proof}
		Let $(\tilde{z}_1, \dots, \tilde{z}_d) \coloneqq \argmin_{(z_1, \dots, z_d) \in \mathcal{Z}}  \sum_{j=1}^{d} a_j h_j(z_j)$. Since $a_j \geq 1$ and $h_j(\cdot)$ is not negative, it holds that
		\begin{align*}
		\min_{(z_1, \dots, z_d) \in \mathcal{Z}}  \sum_{j=1}^{d} a_j h_j(z_j) &= \sum_{j=1}^{d} a_j h_j(\tilde{z}_j) 
		\geq \sum_{j=1}^{d} h_j(\tilde{z}_j) 
		\geq \min_{(z_1, \dots, z_d) \in \mathcal{Z}} \sum_{j=1}^{d} h_j(z_j). \qedhere 
		\end{align*}
	\end{proof}
	
	\begin{lemma} \label{lemma:approx:eps:min:b}
		Let $a_1, \dots, a_d \in [1, \infty[$ and let $h_1, \dots, h_d$ be arbitrary increasing functions. It holds  
		\begin{equation}
		\min_{(z_1, \dots, z_d) \in \mathcal{Z}}  \sum_{j=1}^{d} h_j(a_j z_j) \geq \min_{(z_1, \dots, z_d) \in \mathcal{Z}}  \sum_{j=1}^{d} h_j(z_j).
		\end{equation}
	\end{lemma}
	
	\begin{proof}
		Let $(\tilde{z}_1, \dots, \tilde{z}_d) \coloneqq \argmin_{(z_1, \dots, z_d) \in \mathcal{Z}}  \sum_{j=1}^{d} h_j(a_j z_j)$. Since $a_j \geq 1$ and $h_j$ is an increasing function, it holds that
		\begin{align*}
		\min_{(z_1, \dots, z_d) \in \mathcal{Z}}  \sum_{j=1}^{d} h_j(a_j z_j) &=  \sum_{j=1}^{d} h_j(a_j \tilde{z}_j) 
		\geq \sum_{j=1}^{d} h_j(\tilde{z}_j) +
		\geq \min_{(z_1, \dots, z_d) \in \mathcal{Z}}  \sum_{j=1}^{d} h_j(z_j). \qedhere 
		\end{align*} 
	\end{proof}

	Now, we are able to prove the approximation factor of the operating cost.
	
	\begin{lemma} \label{lemma:approx:eps:op}
		The operating cost of $X'$ is a $(2\gamma - 1)$-approximation, so $C_\text{op}(X') \leq (2\gamma - 1) \cdot C_\text{op}(X^\ast)$.
		
	\end{lemma}
	
	\begin{proof}
		The operating cost of $X'$ is
		\begin{align*}
		C_\text{op}(X') &= \sum_{t=1}^{T} g_t(x'_{t,1}, \dots, x'_{t,d}) \\
		&= \sum_{t=1}^{T} \min_{\substack{(z_1, \dots, z_d) \in \mathcal{Z}}}  \sum_{j=1}^{d} g_{t,j}(x'_{t,j}, z_j) \\
		&= \sum_{t=1}^{T} \min_{\substack{(z_1, \dots, z_d) \in \mathcal{Z}}}  \sum_{\substack{j=1 \\ x'_{t,j} > 0}}^{d}x'_{t,j} f_{t,j} \left( \frac{\lambda_t z_{j}}{x'_{t,j}} \right).
		\end{align*}
		Note that $x'_{t,j} = 0$ implies that $x^\ast_{t,j} = 0$, so the case $x'_{t,j} = 0$ and $\lambda_t z_j > 0$ does not occur. 	
		
		We know that $x'_{t,j} \leq (2\gamma - 1) x^\ast_{t,j}$, so we can use Lemma~\ref{lemma:approx:eps:min:a} with $a_j = (2\gamma - 1) \cdot x^\ast_{t,j} / x_{t,j}$ and $h_j(z_j) = x'_{t,j} f_{t,j} (\lambda_t z_j / x'_{t,j})$. For the special case $x'_{t,j} = 0$ we can simply set $a_j = 1$ and $h_j(z_j) = 0$. Since $x'_{t,j} = 0$ is equivalent to $x^\ast_{t,j} = 0$, we get
		\begin{align*}
		\sum_{t=1}^{T} \min_{\substack{(z_1, \dots, z_d) \in \mathcal{Z}}}  \sum_{\substack{j=1 \\ x'_{t,j} > 0}}^{d} x'_{t,j} f_{t,j} \left( \frac{\lambda_t z_{j}}{x'_{t,j}} \right) 
		\stackrel{L\ref{lemma:approx:eps:min:a}}{\leq}{}
		&
		\sum_{t=1}^{T} \min_{\substack{(z_1, \dots, z_d) \in \mathcal{Z}}}  \sum_{\substack{j=1 \\ x^\ast_{t,j} > 0}}^{d} (2\gamma - 1) x^\ast_{t,j} f_{t,j} \left( \frac{\lambda_t z_{j}}{x'_{t,j}} \right).
		\end{align*}
		
		Furthermore, we can apply Lemma~\ref{lemma:approx:eps:min:b} with $h_j(z_j) = (2\gamma - 1) \cdot x^\ast_{t,j} \cdot \allowbreak f_{t,j} \left( \lambda_t z_j / x'_{t,j} \right)$ and $a_j = x'_{t,j} / x^\ast_{t,j}$. Since $x'_{t,j}  \geq x^\ast_{t,j}$, $a_j \geq 1$ holds and
		\begin{align*}
		\sum_{t=1}^{T} \min_{(z_1, \dots, z_d) \in \mathcal{Z}}  \sum_{\substack{j=1 \\ x^\ast_{t,j} > 0}}^{d} (2\gamma - 1) \cdot x^\ast_{t,j} f_{t,j} \left( \frac{\lambda_t z_{j}}{x'_{t,j}} \right) 
		{}\stackrel{L\ref{lemma:approx:eps:min:b}}{\leq}{}
		& 
		\sum_{t=1}^{T} \min_{\substack{(z_1, \dots, z_d) \in \mathcal{Z}}}  \sum_{\substack{j=1 \\ x^\ast_{t,j} > 0}}^{d} (2\gamma - 1) \cdot x^\ast_{t,j} f_{t,j} \left( \frac{\lambda_t z_{j}}{x^\ast_{t,j}} \right) \\
		\stackrel{\phantom{L\ref{lemma:approx:eps:min:b}}}{=}{}& (2\gamma - 1) \cdot C_\text{op}(X^\ast).
		\end{align*}
		
		Thus, we have a $(2\gamma - 1)$-approximation for the operating cost.
	\end{proof}
	
	Next, we will estimate the switching cost of $X'$:
	
	\begin{lemma} \label{lemma:approx:eps:sw}
		The switching cost of $X'$ is a $(2\gamma - 1)$-approximation, so $C_\text{sw}(X') \leq (2\gamma - 1) \cdot C_\text{sw}(X^\ast)$.
	\end{lemma}
	
	\begin{proof}
		Instead of paying the switching cost for powering up, it is also possible to count the switching costs for powering down (because the first and last state in a schedule is always $\vec{x}_0 = \vec{x}_{T+1} = \vec{0}$ by definition). We will consider the switching cost for each server type separately. For sake of simplicity we will write $x'_{t}$ instead of $x'_{t,j}$ and $x^\ast_t$ instead of $x^\ast_{t,j}$. 
		
		The whole time interval $[T]$ is divided into smaller intervals denoted by $T_1, \dots, T_k$ such that in the \emph{odd} intervals (i.e., $T_1, T_3, \dots$) servers of type $j$ are powered up in $X'$ and in the \emph{even} intervals (i.e., $T_2, T_4, \dots$) servers are powered down. If there are time slots between two intervals where the number of active servers in $X'$ does not change, these time slots belong to the latter interval. Thus, in the last time slot of each interval (except the last one) at least one server is powered up or powered down. 
		Let $t_i$ be the last time slot in $T_i$ and $t_0 \coloneqq 0$, so $T_i = [t_{i-1}+1 : t_i]$.
		
		
		Let $D_{[a:b]}(X) \coloneqq \sum_{t = a}^b (x_{t-1,j} - x_{t,j})^+$ be the number of servers in $X$ that are powered down during the time interval $[a:b]$. For each time interval $I \in \{T_1, \dots, T_k\}$, we have to prove that
		\begin{equation} \label{eqn:approx:eps:sw:intervalcost}
		D_I(X') \leq (2\gamma - 1) \cdot D_I(X^\ast).
		\end{equation} 
		In the odd intervals no servers are powered down, so the inequality is always satisfied.
		Let $T_i = [t_{i-1}+1:t_i] \in \{T_2, T_4, \dots\}$ be an even interval. During time slot $t_{i-1}$ (this is the last time slot of the previous interval) servers were powered up in $X'$. By the definition of $X'$ (see equation~\eqref{eqn:approx:eps:rule}), this implies that 
		\begin{equation} \label{eqn:approx:eps:sw:before}
		x^\ast_{t_{i-1}} > x'_{t_{i-1}} / \gamma
		\end{equation}
		because otherwise we had not reach the state $x'_{t_{i-1}}$. 
		
		The last state in the time interval $T_i$ is $x'_{t_i}$. Of course, this state satisfies the invariant~\eqref{eqn:approx:eps:invariant}.
		In contrast, the next larger state in $M^{\gamma}_j$ denoted by $x^+_{t_i} \coloneqq N_j(x'_{t_i})$ does not satisfy it, so
		\begin{equation} \label{eqn:approx:eps:sw:after}
		x^+_{t_i} > (2\gamma - 1) \cdot x^\ast_{t_i}
		\end{equation}
		because otherwise the last state of $T_i$ would be $x^+_{t_i}$.
		By the definition of $M^{\gamma}_j$, the relative distance between $x'_{t_i}$ and $x^+_{t_i}$ is at most $\gamma$, i.e.,
		\begin{equation} \label{eqn:approx:eps:sw:plusdist}
		\frac{x^+_{t_i}}{\gamma} \leq x'_{t_i}.
		\end{equation}
		Furthermore, we know that
		\begin{equation} \label{eqn:approx:eps:sw:plusbegin}
		x'_{t_{i-1}} \geq x^+_{t_i}
		\end{equation}
		because during the time interval at least one server is powered down. 
		
		By using the inequalities~\eqref{eqn:approx:eps:sw:before}, \eqref{eqn:approx:eps:sw:after}, \eqref{eqn:approx:eps:sw:plusdist} and~\eqref{eqn:approx:eps:sw:plusbegin} as well as $\gamma > 1$, we can prove~\eqref{eqn:approx:eps:sw:intervalcost}:
		\begin{align*}
		(2\gamma - 1) \cdot D_{[t_{i-1}+1:t_i]}(X^\ast) &\geq (2\gamma - 1)\cdot (x^\ast_{t_{i-1}} - x^\ast_{t_i}) \\
		&\alignstack{\eqref{eqn:approx:eps:sw:before}}{>} (2\gamma - 1)\cdot \frac{x'_{t_{i-1}}}{\gamma} - (2\gamma - 1) \cdot x^\ast_{t_i} \\
		&\alignstack{\eqref{eqn:approx:eps:sw:after}}{>}  (2\gamma - 1)\cdot \frac{x'_{t_{i-1}}}{\gamma} - x^+_{t_i}  \\
		&\alignstack{\eqref{eqn:approx:eps:sw:plusdist}}{\geq}  (2\gamma - 1)\cdot \frac{x'_{t_{i-1}}}{\gamma} - x^+_{t_i}  + \frac{x^+_{t_i}}{\gamma} - x'_{t_i}\\
		&= x'_{t_{i-1}} - x'_{t_i} + (1 - 1/\gamma) (x'_{t_{i-1}} - x^+_{t_i}) \\
		&\alignstack{\eqref{eqn:approx:eps:sw:plusbegin}}{>} x'_{t_{i-1}} - x'_{t_i} \\
		&= D_{[t_{i-1}+1:t_i]}(X'). \numberthis \label{eqn:approx:eps:sw:intervalcost:proven}
		\end{align*}
		The last inequality holds, because $(1 - 1/\gamma) > 0$ and $x'_{t_{i-1}} - x^+_{t_i} \geq 0$.
		
		For the whole workload and all server types, we get
		\begin{align*}
		C_\text{sw}(X') 
		&= \sum_{j=1}^{d} \beta_j \cdot  \sum_{i=1}^{k} D_{T_i}(X') \\
		&\alignstack{\eqref{eqn:approx:eps:sw:intervalcost:proven}}{>} \sum_{j=1}^{d} \beta_j \cdot  \sum_{i=1}^{k} (2\gamma - 1) D_{T_i}(X^\ast) \\
		&= (2\gamma - 1) \cdot C_\text{sw}(X^\ast),
		\end{align*} 
		so the schedule $X'$ is a $(2\gamma - 1)$-approximation according to the switching cost. \qedllncs
	\end{proof}
	
	Finally, we can prove theorem~\ref{theo:approx:eps:theorem}:
	
	\begin{proof}[\textbf{Proof of Theorem~\ref{theo:approx:eps:theorem}}]
		By lemma~\ref{lemma:approx:eps:op} and~\ref{lemma:approx:eps:sw}, it holds
		\begin{align*}
		C(X^{\gamma}) &\leq C(X') \\
		&= C_\text{op}(X') + C_\text{sw}(X') \\
		&\leq (2\gamma - 1) \cdot C_\text{op}(X^\ast)  + (2\gamma - 1) \cdot C_\text{sw}(X^\ast) \\
		&= (2\gamma - 1) \cdot C(X^\ast).
		\end{align*}
		Therefore, $X^{\gamma}$ is a $(2\gamma - 1)$-approximation.\qedllncs
	\end{proof}
	

If we set $\gamma = (1+\epsilon / 2)$, then $2\gamma -1 = (1+\epsilon)$, so we have a $(1+\epsilon)$-approximation. For server type~$j$, there are $|M^{\gamma}_j| \in \mathcal{O}(\log_{\gamma} m_j) =  \mathcal{O}(\log_{1+\epsilon} m_j)$ different values that are used by our graph-based algorithm. Thus the graph consists of
\begin{equation*}
\mathcal{O}\left(T \cdot \prod_{j=1}^{d} \log_{1+\epsilon} m_j\right)
\end{equation*}
vertices which is also the algorithm's runtime. 
For $\epsilon < 1$ (usually we are not interested in $\epsilon$-values that are bigger than~1) the term $\frac{1}{\log (1 + \epsilon)}$ can be written as $1/\epsilon + O(1)$, so the runtime is
$\mathcal{O}\left(T \cdot \epsilon^{-d} \cdot \prod_{j=1}^{d} \log m_j\right)$.
We summarize our results in the following theorem:

\begin{theorem} \label{theo:approx:eps:time}
	Given the problem instance $\mathcal{I}$, a $(1+\epsilon)$-\hspace{0pt}approximation can be calculated in 
	\begin{equation*}
	\mathcal{O}\left(T \cdot \epsilon^{-d} \cdot \prod_{j=1}^{d} \log m_j\right)
	\end{equation*}
	 time. 
\end{theorem}

\subsection{Time-dependent data-center size}
\label{sec:approx:mtj}

In practice, the size of a data center can change over time. If a data center is extended with new servers of type $j$, then $m_j$ increases. If parts of the data center are shut down for maintenance, $m_j$ decreases temporarily. Let $m_{t,j}$ denote the total number of servers of type~$j$ at time slot $t$. In the following, we will show that the approximation algorithm still works in this setting. 

Let $M_{t,j} \coloneqq [m_{t,j}]_0$ and $\mathcal{M}_t = \varprod_{j=1}^d M_{t,j}$ be the allowed server configurations at time slot~$t$. The vertices in $G$ that represent unavailable server configurations are removed along with the incident edges. The shortest path in the new graph represents an optimal schedule. For the approximation, let 
\begin{align*}
M^\gamma_{t,j} \coloneqq \{0, m_{t,j}\} \cup \{\lfloor \gamma^k \rfloor \in M_{t,j} \mid k \in \mathbb{N} \} \cup \{\lceil \gamma^k \rceil \in M_{t,j} \mid k \in \mathbb{N} \}
\end{align*}
and let $\mathcal{M}^\gamma_t \coloneqq \varprod_{j=1}^d M^\gamma_{t,j}$ be the considered server configurations. The resulting graph is denoted by~$\bar{G}^\gamma$. Theorem~\ref{theo:approx:eps:theorem} still hold for the modified graph, i.e., the schedule that corresponds to the shortest path in $\bar{G}^\gamma$ is a $(2 \gamma - 1)$-approximation. The following theorem shows that a $(1 + \epsilon)$-approximation can still be calculated in polynomial time (if $d$ is a constant). The proof is analogous to that of Theorem~\ref{theo:approx:eps:time}.

\begin{theorem}
	Given the problem instance $\mathcal{I}$ where the total number of available servers depends on time, a $(1 + \epsilon)$-approximation can be calculated in 
	\begin{equation*}
	\mathcal{O}\left(\epsilon^{-d} \cdot \sum_{t=1}^{T} \prod_{j=1}^{d} \log m_{t,j}\right) \subseteq \mathcal{O}\left(T \cdot \epsilon^{-d} \cdot \prod_{j=1}^d \log \max_{t \in [T]} m_{t,j}\right)
	\end{equation*} 
	time.
\end{theorem}

%% file: image_graph.tex
\begin{figure}[tb]

	\centering
	\begin{tikzpicture}
		\pgfmathsetmacro{\xBegin}{0}
		\pgfmathsetmacro{\xStep}{3} 
		\pgfmathsetmacro{\xTimeGap}{0.5} 
		
		\pgfmathsetmacro{\yBegin}{0}
		\pgfmathsetmacro{\yStep}{2.2} 
		\pgfmathsetmacro{\yMax}{2}
		\pgfmathsetmacro{\yMaxMinusOne}{\yMax - 1}
		
		\pgfmathsetmacro{\zxStep}{1.2}
		\pgfmathsetmacro{\zyStep}{0.9}
		\pgfmathsetmacro{\zMax}{1}
		\pgfmathsetmacro{\zMaxMinusOne}{\zMax - 1}
		
		\pgfmathsetmacro{\xLength}{2}
		
		\tikzstyle{textsize}=[font=\scriptsize]
		\tikzstyle{vertex}=[fill, circle, inner sep=1pt, outer sep=1.5pt]
		\tikzstyle{vertextext}=[textsize, above left, xshift=1.5pt, yshift=-1.5pt]
		\tikzstyle{opEdgeNode}=[textsize, below, yshift=1.5pt]
		\tikzstyle{swEdgeNodeY}=[textsize, left, xshift=1.5pt]
		\tikzstyle{swEdgeNodeZ}=[textsize, above left, yshift=-2.5pt, xshift=2.5pt]
		\tikzstyle{arrow}=[-stealth]
		\tikzstyle{arrowpath}=[-stealth, green, thick]
		
		\foreach \y in {0,...,\yMax} {
			\foreach \z in {0,...,\zMax} {
				\node[vertex] (n0y\y z\z) at 
					(\xBegin + 0 * \xStep + \z * \zxStep,
					 \yBegin + \y * \yStep + \z * \zyStep) {};
				\node[vertex] (n1y\y z\z) at 
					(\xBegin + 1 * \xStep + \z * \zxStep,
					 \yBegin + \y * \yStep + \z * \zyStep) {};
				\node[vertex] (n2y\y z\z) at 
					(\xBegin + 2 * \xStep + \z * \zxStep + \xTimeGap,
					 \yBegin + \y * \yStep + \z * \zyStep) {};
				\node[vertex] (n3y\y z\z) at 
					(\xBegin + 3 * \xStep + \z * \zxStep + \xTimeGap,
					 \yBegin + \y * \yStep + \z * \zyStep) {};
				\node[vertextext] at (n0y\y z\z)
					{$v^\uparrow_{1,(\y,\z)}$};
				\node[vertextext] at (n1y\y z\z)
					{$v^\downarrow_{1,(\y,\z)}$};
				\node[vertextext] at (n2y\y z\z)
					{$v^\uparrow_{2,(\y,\z)}$};
				\node[vertextext] at (n3y\y z\z)
					{$v^\downarrow_{2,(\y,\z)}$};
				
				\draw[arrow] (n0y\y z\z) -- (n1y\y z\z)
					node[opEdgeNode, pos=0.65-0.3*\z] {$g_1(\y, \z)$};
				\draw[arrow] (n1y\y z\z) -- (n2y\y z\z)
					node[opEdgeNode, pos=0.65-0.3*\z] {$0$};
				\draw[arrow] (n2y\y z\z) -- (n3y\y z\z)
					node[opEdgeNode, pos=0.65-0.3*\z] {$g_2(\y, \z)$};
			}
		}
	
		\foreach \x in {0,2} {
			\foreach \y in {0,...,\yMaxMinusOne} {
				\foreach \z in {0,...,\zMax} {
					\pgfmathsetmacro{\yp}{int(\y + 1)}
					\pgfmathsetmacro{\zp}{int(\z + 1)}
					\pgfmathsetmacro{\xp}{int(\x + 1)}
					\draw[arrow] (n\x y\y z\z) -- (n\x y\yp z\z) 
						node[swEdgeNodeY, pos=0.6-0.2*\z] {$\beta_1$};
					\draw[arrow] (n\xp y\yp z\z) -- (n\xp y\y z\z) 
						node[swEdgeNodeY, pos=0.4+0.2*\z] {$0$};
				}
			}
		}
	
		\foreach \x in {0,2} {
			\foreach \y in {0,...,\yMax} {
				\foreach \z in {0,...,\zMaxMinusOne} {
					\pgfmathsetmacro{\yp}{int(\y + 1)}
					\pgfmathsetmacro{\zp}{int(\z + 1)}
					\pgfmathsetmacro{\xp}{int(\x + 1)}
					\draw[arrow] (n\x y\y z\z) -- (n\x y\y z\zp)
						node[swEdgeNodeZ, midway, ] {$\beta_2$};
					\draw[arrow] (n\xp y\y z\zp) -- (n\xp y\y z\z)
						node[swEdgeNodeZ, midway] {$0$};
				}
			}
		}
	
	\draw[arrowpath] (n0y0z0) -- (n0y1z0);
	\draw[arrowpath] (n0y1z0) -- (n0y2z0);
	\draw[arrowpath] (n0y2z0) -- (n1y2z0);
	\draw[arrowpath] (n1y2z0) -- (n1y1z0);
	\draw[arrowpath] (n1y1z0) -- (n2y1z0);
	\draw[arrowpath] (n2y1z0) -- (n2y1z1);
	\draw[arrowpath] (n2y1z1) -- (n3y1z1);
	\draw[arrowpath] (n3y1z1) -- (n3y0z1);
	\draw[arrowpath] (n3y0z1) -- (n3y0z0);
	
	\node[fill=red, circle, inner sep=1.8pt] at (n0y0z0) {};
	\node[fill=blue, circle, inner sep=1.8pt] at (n3y0z0) {};

	\end{tikzpicture}
	\caption{{\normalfont (This figure is colored)} Visualization of the graph representation. This example shows two server types ($d=2$) and two time slots ($T=2$). There are $m_1 = 2$ servers of type~1 and $m_2 = 1$ server of type~2. The algorithm calculates a shortest path from $v^\uparrow_{1,(0,0)}$ (red dot) to $v^\downarrow_{2,(0,0)}$ (blue dot). The shortest path is drawn in green and corresponds to the optimal schedule $\vec{x}_1 = (2,0)$ and $\vec{x}_2 = (1,1)$.}
	\label{fig:approx:graph}
\end{figure}

%% file: image_approxXprime.tex
\begin{figure}[tb]

	\centering
	\begin{tikzpicture}
		\pgfmathsetmacro{\xBegin}{0}
		\pgfmathsetmacro{\xStep}{0.58} 
		\pgfmathsetmacro{\xLength}{18}
		\pgfmathsetmacro{\xMax}{\xLength - 1}
		
		\pgfmathsetmacro{\yBegin}{0}
		\pgfmathsetmacro{\yStep}{0.4} 
		\pgfmathsetmacro{\yXAxis}{\yBegin - 0.5 * \yStep}
		\pgfmathsetmacro{\yMax}{10}
		\pgfmathsetmacro{\yOptOffset}{-0.00}
		\pgfmathsetmacro{\yApproxOffset}{0.00}
		\pgfmathsetmacro{\yUpperOffset}{+0.00}
		
		\pgfmathsetmacro{\axisExtend}{0.7}
		\pgfmathsetmacro{\tickHalfLength}{0.07}
		
		\def\yvalues{0,1,2,4,8,10}
		\def\xopt{{0,1,3,5,7,5,6,3,2,1,2,3,9,4,3,2,1,0,0}}
		\def\xapprox{{0,1,4,8,8,8,8,8,4,2,2,4,10,10,8,4,2,0,0}}
		
		\tikzstyle{axis}=[-stealth]
		\tikzstyle{allowedstates}=[dashed]
		\tikzstyle{opt}=[thin,red]
		\tikzstyle{upper}=[dotted,blue]
		\tikzstyle{approx}=[ultra thick,green!50!black]
		
		\draw[axis] (\xBegin, \yXAxis) 
			to (\xBegin + \xStep * \xLength + \axisExtend * \xStep, \yXAxis) node[below] {$t$};
		\draw[axis] (\xBegin, \yBegin) 
			to (\xBegin, \yBegin + \yStep * \yMax + \axisExtend * \yStep) node[left, yshift=1pt] {$x_{t,j}$};
		\foreach \x in {0,...,\xLength} {
			\draw[-] 
				(\xBegin + \x * \xStep, \yXAxis + \tickHalfLength) -- 
				(\xBegin + \x * \xStep, \yXAxis - \tickHalfLength);
		}
		\foreach \x in {0,...,\xMax} {
			\node[below] at 
				(\xBegin + \x * \xStep + 0.5 * \xStep,	\yXAxis) 
				{$\x$};
		}
		\foreach \y in {0,...,\yMax} {
			\draw[-] 
				(\xBegin + \tickHalfLength, \yBegin + \y * \yStep) -- 
				(\xBegin - \tickHalfLength, \yBegin + \y * \yStep);
		}
		\foreach \y in {0,...,\yMax} {
			\draw[-] 
				(\xBegin + \tickHalfLength, \yBegin + \y * \yStep) -- 
				(\xBegin - \tickHalfLength, \yBegin + \y * \yStep);
		}
		\foreach \y in \yvalues {
			\node[left] at 
				(\xBegin - \tickHalfLength, \yBegin + \y * \yStep) {$\y$};
			\draw[allowedstates] (\xBegin, \yBegin + \y * \yStep)
				-- (\xBegin + \xLength * \xStep, \yBegin + \y * \yStep);
		}
		
		\foreach \x in {0,...,\xMax} {
			
			\pgfmathsetmacro{\xp}{\x + 1}
			\pgfmathsetmacro{\yfac}{\xapprox[\x]}
			\pgfmathsetmacro{\yfacp}{\xapprox[\xp]}
			\draw[approx] 
				(\xBegin + \x * \xStep, \yBegin + \yfac * \yStep + \yApproxOffset)
				-- (\xBegin + \xp * \xStep, \yBegin + \yfac * \yStep + \yApproxOffset)
				-- (\xBegin + \xp * \xStep, \yBegin + \yfacp * \yStep + \yApproxOffset)
				-- (\xBegin + \xp * \xStep + 0.01, \yBegin + \yfacp * \yStep + \yApproxOffset);
		}
		\foreach \x in {0,...,\xMax} {
			
			\pgfmathsetmacro{\xp}{\x + 1}
			\pgfmathsetmacro{\yfac}{min(\yMax,3*\xopt[\x])}
			\pgfmathsetmacro{\yfacp}{min(\yMax,3*\xopt[\xp])}
			\draw[upper] 
				(\xBegin + \x * \xStep, \yBegin + \yfac * \yStep + \yUpperOffset)
				-- (\xBegin + \xp * \xStep, \yBegin + \yfac * \yStep + \yUpperOffset);
			\draw[upper]
				(\xBegin + \xp * \xStep, \yBegin + \yfac * \yStep + \yUpperOffset)
				-- (\xBegin + \xp * \xStep, \yBegin + \yfacp * \yStep + \yUpperOffset);
		}
		\foreach \x in {0,...,\xMax} {
			
			\pgfmathsetmacro{\xp}{\x + 1}
			\pgfmathsetmacro{\yfac}{\xopt[\x]}
			\pgfmathsetmacro{\yfacp}{\xopt[\xp]}
			\draw[opt] 
				(\xBegin + \x * \xStep, \yBegin + \yfac * \yStep + \yOptOffset)
				-- (\xBegin + \xp * \xStep, \yBegin + \yfac * \yStep + \yOptOffset)
				-- (\xBegin + \xp * \xStep, \yBegin + \yfacp * \yStep + \yOptOffset)
				-- (\xBegin + \xp * \xStep + 0.01, \yBegin + \yfacp * \yStep + \yOptOffset);
		}
		
	\end{tikzpicture}
	\caption{
		{\normalfont (This figure is colored)} 
		Visualization of the construction of $X'$ (shown in green) for one specific server type $j$. In this example, we have $\gamma = 2$ and $m_j = 10$, so the allowed states for $X'$ are $M^\gamma_j = \{0,1,2,4,8,10\}$ (dashed horizontal lines). The optimal schedule $X^\ast$ is shown in red. The dotted blue line shows the value of $\min(m_j, (2\gamma - 1) x^\ast_{t,j})$. Note that the schedule $X'$ always stays between the red and the blue line and only changes the number of active servers to ensure the invariant.
	}
	\label{fig:approx:xprime}
	
\end{figure}
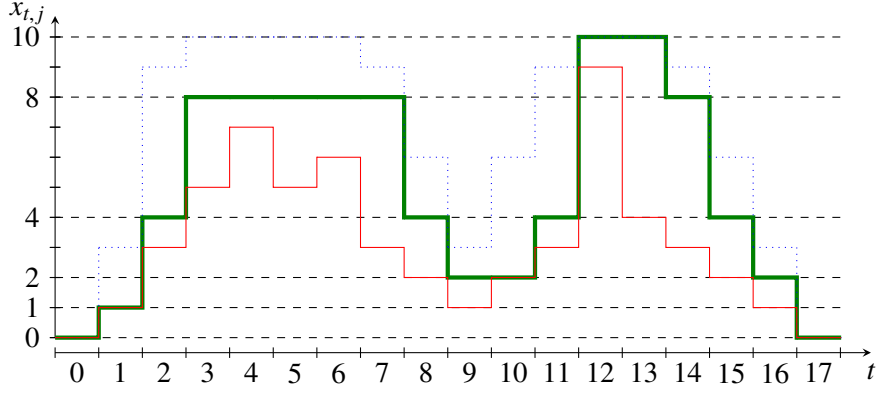

%% file: appendixVariables.tex
\newpage
\section{Variables and notation}
\label{sec:appendix:variables}

The following table gives an overview of the variables defined in this paper. 

\begingroup
	\setlength{\tabcolsep}{6pt}
	\renewcommand{\arraystretch}{1.3}
	\small
	
	\centering
	\begin{longtable}{|l|p{13.0cm}|}
		\hline
		\textbf{Variable} & \textbf{Description} \\
		\hline
		\endhead
		\hline
		\endfoot
		$A_{j,i}$ & Block that contains the time slots when a server of type $j$ is in the active state. Formally, $A_{j,i} \coloneqq [s_{j,i} : s_{j,i} + \bar{t}_j - 1]$ (in Section~\ref{sec:online}) and $A_{j,i} \coloneqq [s_{j,i} : s_{j,i} + \bar{t}_{t,j}]$ (in Section~\ref{sec:online:time}). \\ 
		$\mathcal{A}$ & Our online algorithm for time-independent operating cost functions $f_j$ (see Section~\ref{sec:online}). \\
		$B_{j,k}$ & Set of the indices $i$ of the blocks $A_{j,i}$ that contain the special time slot $\tau_{j,k}$. Formally, $B_{j,k} \coloneqq \{i \in [n_j] \mid A_{j,i} \ni \tau_{j,k}\}$. An example is shown in Figure~\ref{fig:online:tau}. \\
		$\mathcal{B}$ & Our online algorithm for time-dependent operating cost functions $f_{t,j}$ achieving a competitive ratio of $2d + 1 + c(\mathcal{I})$ (see Section~\ref{sec:online:time:ci}). \\
		$\beta_j$ & Switching cost of server type $j$. \\
		$c(\mathcal{I})$ & Constant depending on the problem instance $\mathcal{I}$, $c(\mathcal{I}) \coloneqq \sum_{j = 1}^{d}\max_{t \in [T]} {l_{t,j}}/{\beta_j}$. \\ 
		$C(X)$ & Total cost of the schedule $X$ (see equation~\eqref{eqn:problem:costx}). \\
		$C_I(X)$ & Total cost of the schedule $X$ during the time interval $I$. \newline Formally, $C_I(X) \coloneqq \sum_{t \in I} \left( g_t(\vec{x}_t) + \sum_{j=1}^{d} \beta_j (x_{t,j} - x_{t-1,j})^+ \right)$. \\
		$C_\text{op}(X), C^\mathcal{J}_\text{op}(X)$ & Operating cost of the schedule $X$ (regarding the problem instance $\mathcal{J}$), $C_\text{op}(X) \coloneqq \sum_{t=1}^{T} g_t(\vec{x}_t)$. \\
		$C_\text{sw}(X), C^\mathcal{J}_\text{sw}(X)$ & Switching cost of the schedule $X$  (regarding the problem instance $\mathcal{J}$), \newline $C_\text{sw}(X) \coloneqq \sum_{t=1}^{T} \sum_{j=1}^{d} \beta_j (x_{t,j} - x_{t-1,j})^+$. \\
		$\mathcal{C}$ & Our online algorithm for time-dependent operating cost function $f_{t,j}$ achieving a competitive ratio of $2d + 1 + \epsilon$ (see Section~\ref{sec:online:time:epsilon}). \\
		$\gamma$ & Parameter used for the approximation algorithm. The ratio between two consecutive states is at most~$\gamma$. \\
		$d$ & Number of server types. \\
		$f_{t,j}(z)$ & Operating cost of a single server of type $j$ running with load $z \in [0, z^\text{max}_j]$ at time slot $t$. \\
		$g_{t,j}(x,z)$ & Operating cost for $x$ servers of type $j$ processing the fraction $z$ of the job volume $\lambda_t$ at time slot $t$. \\ 
		$g_t(\vec{x}_t)$ & Operating cost during time slot $t$ for the server configuration $\vec{x}_t$, see equation~\eqref{eqn:problem:gt}. \\
		$\tilde{g}_u(\vec{x}_u)$ & Operating cost in the modified problem instance $\tilde{\mathcal{I}}$ during time slot $u$ for the server configuration $\vec{x}_t$. \\
		$G, G(\mathcal{I})$ & Graph used for the optimal offline algorithm. \\
		$G^\gamma$ & Graph used for the approximation algorithm with parameter $\gamma$. \\
		$H_{j,i}$ & Switching and idle operating cost of block $A_{j,i}$ (see equation~\eqref{lemma:online:func:dji} for algorithm $\mathcal{A}$ or equation~\eqref{eqn:online:time:hji} for algorithm~$\mathcal{B}$). \\
		$\mathcal{I}$ & Problem instance. Formally,  $\mathcal{I} \coloneqq (T, d, \vec{m}, \vec{\beta}, F, \Lambda)$. \\
		$\mathcal{I}^t$ & Problem instance that ends at time slot $t$. Formally, $\mathcal{I}^t \coloneqq (t, d, \vec{m}, \vec{\beta}, F, \Lambda^t)$. \\
		$\tilde{\mathcal{I}}$ & Modified problem instance, see Section~\ref{sec:online:time:epsilon}, $\tilde{\mathcal{I}} \coloneqq (\tilde{T}, d, \vec{m}, \vec{\beta}, \tilde{F}, \tilde{\Lambda})$. \\
		$l_{t,j}$ & Idle operating cost of server type $j$ at time slot $t$ (Section~\ref{sec:online:time}), $l_{t,j} \coloneqq f_{t,j}(0)$. \\
		$\tilde{l}_{u,j}$ & Idle operating cost of server type $j$ at time slot $u$ in the modified problem instance $\tilde{\mathcal{I}}$ (Section~\ref{sec:online:time:epsilon}), $\tilde{l}_{u,j} \coloneqq \tilde{f}_{u,j}(0)$. \\ 
		$L_{t,j}(X)$ & Load-dependent operating cost of all servers of type $j$ at time slot $t$ in the schedule $X$. \\
		$\lambda_t$ & Job volume that arrives at time slot $t$. \\
		$m_j$ & Number of servers of type $j$. \\
		$M_j$ & $M_j \coloneqq [m_j]_0$. \\
		$M^\gamma_j$ & Possible numbers of active servers of type $j$ in the approximation algorithm with parameter $\gamma$. Formally, $M^\gamma_j \coloneqq \{0, 1, \lfloor \gamma^1 \rfloor, \lceil \gamma^1 \rceil, \lfloor \gamma^2 \rfloor, \lceil \gamma^2 \rceil, \dots, m_j\}$. \\
		$\mathcal{M}$ & Set of all possible server configurations, $\mathcal{M} \coloneqq \bigtimes_{j=1}^d M_j$. \\
		$\mathcal{M}^\gamma$ & Set of all possible server configurations in the approximation algorithm, $\mathcal{M}^\gamma \coloneqq \bigtimes_{j=1}^d M^\gamma_j$. \\
		$\mu(t)$ & Time slot in $U(t)$ where the operating cost in the schedule $X^\mathcal{B}$ is minimal, $\mu(t) \coloneqq \argmin_{u \in U(t)} \tilde{g}_{u}({\vec{x}}^\mathcal{B}_{u})$. \\
		$n_j$ & Number of blocks for server type $j$, the variables $A_{j,i}, s_{j,i}$ and $H_{j,i}$ are defined for $i \in [n_j]$. \\
		$n'_j$ & Number of special time slots for server type $j$, the variables $\tau_{j,k}$ and $B_{j,k}$ are defined for $k \in [n'_j]$. \\ 
		$\tilde{n}_t$ & Each time slot $t$ in the original problem instance $\mathcal{I}$ is divided into $\tilde{n}_t$ time slots in the modified problem instance $\tilde{\mathcal{I}}$. \\
		$N_j(x_j)$ & The next greater value of $x_j$ in $M^\gamma_j$. Formally, $N_j(x_j) \coloneqq \min \{x \in M^{\gamma}_j \mid x > x_j\}$. \\
		$s_{j,i}$ & Time slot when a server of type $j$ is powered up in $X^\mathcal{A}$ or $X^\mathcal{B}$. It holds $s_{j,1} \leq \dots \leq s_{j, n_j}$. \\
		$\bar{t}_j$ & Number of time slots that a server of type $j$ stays active in algorithm~$\mathcal{A}$ (\emph{including} the time slot when the server was powered up); $\bar{t}_j \coloneqq \left\lceil \beta_j / l_j \right\rceil$. \\
		$\bar{t}_{t,j}$ & Number of time slots that a server of type $j$ stays active in algorithm~$\mathcal{B}$ (\emph{excluding} the time slot when the server was powered up), $\bar{t}_{t,j} \coloneqq \max \{\bar{t} \in [T-t] \mid \sum_{u = t+1}^{t+\bar{t} } l_{u,j} \leq \beta_j\}$. \\
		$\tau_{j,k}$ & Special time slots. Each block $A_{j,i}$ ($j \in [d]$, $i \in [n_j]$) contains exactly one special time slot $\tau_{j,k}$ with $k \in [n'_j]$. An example is shown in Figure~\ref{fig:online:tau}. \\
		$T$ & Total number of time slots. \\
		$\tilde{T}$ & Total number of time slots in the modified problem instance $\tilde{\mathcal{I}}$. \\
		$U(t)$ & Set of time slots in the modified problem instance $\tilde{\mathcal{I}}$ that correspond to the time slot $t$ of the original problem instance $\mathcal{I}$. \\
		$U^{-1}(u)$ & Time slot in the original problem instance $\mathcal{I}$ that corresponds to the time slot $u$ of the modified problem instance $\tilde{\mathcal{I}}$. \\
		$X$ & An arbitrary schedule. Formally, $X = (\vec{x}_1, \dots, \vec{x}_T)$ and $\vec{x}_t = (x_{t,1}, \dots, x_{t,d})$. \\
		$X^\ast$ & An optimal schedule. \\
		$X^\mathcal{A}, X^\mathcal{B}, X^\mathcal{C}$ & The schedule calculated by our online algorithm~$\mathcal{A}$,~$\mathcal{B}$ and~$\mathcal{C}$, see Sections~\ref{sec:online}, \ref{sec:online:time:ci} and \ref{sec:online:time:epsilon}, respectively. \\
		$\hat{X}^t$ & An optimal schedule for the problem instance $\mathcal{I}^t$ that ends at time $t$. \\
		$x_{t,j}$ & Number of active servers of type $j$ at time $t$ in the schedule $X$. \\
		$\vec{x}_t$ & Server configuration at time slot $t$ in schedule $X = (\vec{x}_1, \dots, \vec{x}_T)$, $\vec{x}_t = (x_{t,1}, \dots, x_{t,d})$. \\
		$x^\mathcal{A}_{t,j},x^\mathcal{B}_{t,j},x^\mathcal{C}_{t,j}$ & Number of active servers of type $j$ at time $t$ in the schedule $X^\mathcal{A}$, $X^\mathcal{B}$ and $X^\mathcal{C}$, respectively. \\
		$\hat{x}^u_{t,j}$ & Number of active servers of type $j$ at time $t$ in the schedule $\hat{X}^u$. \\
		$w_{t,j}$ & Number of servers of type $j$ that were powered up by our online algorithm at time slot $t$. \\
		$W_t$ & Set of time slots $u$ with $u + \bar{t}_{u,j} + 1 = t$. Servers that were powered up at time slot $u \in W_t$ in algorithm~$\mathcal{B}$ will be powered down at time slot $t$. See Figure~\ref{fig:online:time} for an example. \\
		$z^{\text{max}}_j$ & Maximum job volume that can be processed by one server of type $j$ during a single time slot. \\
		$z_{t,j}$ & Ratio of the job volume $\lambda_t$ that is processed by server type $j$. \\ 
		$\mathcal{Z}$ & Set of all possible job assignments, $\mathcal{Z} \coloneqq \big\{(z_1, \dots, z_d) \in [0,1]^d \mid \sum_{j=1}^{d} z_j = 1 \big\}$. \\
	\end{longtable}
\endgroup